\documentclass[%
 reprint,
superscriptaddress,
]{revtex4-2}
\usepackage[utf8]{inputenc}
\usepackage{graphicx}
\usepackage{dcolumn}
\usepackage{dsfont}
\usepackage{amssymb}
\usepackage{MnSymbol}
\usepackage{amsfonts}
\usepackage{mathtools}
\usepackage{amsmath}
\usepackage{amsthm}
\usepackage{xcolor}
\usepackage{hyperref}
\usepackage{comment}
\usepackage{physics}
\usepackage{soul}
\usepackage{tikz}
\usetikzlibrary{shapes.geometric}

\newcommand\mc[1]{\mathcal{#1}}

\newcommand{\cH}{{\mc{H}}}
\newcommand{\cG}{{\mc{G}}}

\newcommand{\cP}{{\mc{P}}}

\newcommand{\aP}{{\widetilde{\mc{P}}}}
\newcommand{\cJ}{{J}}
\newcommand{\cS}{{\mc{S}}}

\newcommand{\ta}{{\widetilde{a}}}

\newcommand{\tA}{{\widetilde{A}}}

\newcommand{\tB}{{\widetilde{B}}}

\newcommand{\tP}{{\widetilde{P}}}
\newcommand{\tQ}{{\widetilde{Q}}}

\newcommand{\cU}{{\mc{U}}}

\newcommand{\R}{{\mathbb R}}
\newcommand{\C}{{\mathbb C}}
\newcommand{\CP}{{\mathbb{CP}}}

\newcommand{\Cl}{{\mathcal{C}}}
\newcommand{\one}{{\mathds{1}}}

\newcommand{\zz}{{\mathbb{Z}}}
\newcommand{\ra}{{\rightarrow}}

\newcommand{\tot}{{\mathrm{tot}}}

\newcommand{\Sp}{{\mathrm{Sp}}}
\newcommand{\fg}{{L}}
\newcommand{\PI}{{\cS}}

\newcommand{\pentagon}{%
  \mathbin{%
    \tikz[baseline=-0.55ex]%
      \node[
        draw,
        regular polygon,
        regular polygon sides=5,
        minimum size=1.2ex,
        inner sep=0pt,
        line width=0.3pt
      ] {};
  }%
}
\newcommand{\pent}{{\cS^{\pentagon}}}
\newcommand{\pentone}{{\cS^{\pentagon}_\one}}

\newcommand{\cx}[1]{{\textcolor{green!70!black}{{#1}}}}
\newcommand{\cz}[1]{{\textcolor{red}{{#1}}}}
\newcommand{\cw}[1]{{\color{red}{{#1}}}}

\newtheorem{definition}{Definition}
\newtheorem{lemma}{Lemma}

\newtheorem{proposition}{Proposition}
\newtheorem{theorem}{Theorem}
\newtheorem*{theorem*}{Theorem}

\begin{document}
\preprint{APS/123-QED}

\title{An algebraically closed family of informational $n$-qubit purity invariants}
\author{Markus Frembs}
\email{markus.frembs@itp.uni-hannover.de}
\thanks{(first author)}
\affiliation{Institut f\"ur Theoretische Physik, Leibniz Universit\"at Hannover, Appelstraße 2, 30167 Hannover, Germany\\ Okinawa Institute of Science and Technology Graduate University, 1919-1 Tancha, Onna-son, Okinawa 904-0495, Japan}
\author{Giovanni Natale}
\email{giovanninatale96@gmail.com}
\affiliation{Okinawa Institute of Science and Technology Graduate University, 1919-1 Tancha, Onna-son, Okinawa 904-0495, Japan}
\author{Christopher S.P. Wever}
\email{christopher.wever@de.bosch.com}
\affiliation{Corporate Sector Research and Advance Engineering, Robert Bosch GmbH, Robert-Bosch-Campus 1, D-71272 Renningen, Germany}
\author{Philipp A. H\"ohn}
\email{philipp.hoehn@oist.jp}
\affiliation{Okinawa Institute of Science and Technology Graduate University, 1919-1 Tancha, Onna-son, Okinawa 904-0495, Japan}

\begin{abstract}
    We present a family of quadratics in Pauli expectation values, and prove that they constitute state-independent invariants for all $n$-qubit pure states. This family generalises the two-qubit `pentagon identities', discovered in the reconstruction programme of Ref.~\cite{Hoehn2017,HoehnWever2017}, where they characterise the space of pure states, as well as the unitary group, and are interpreted as complementarity equalities in the Brukner-Zeilinger information measure. The generalisation to arbitrarily many qubits is nontrivial as it requires new tools which in turn reveal novel structural properties that are absent in the two-qubit case. A thorough analysis of these properties, and their relation with mutual unbiasedness and complementarity  in the $n$-qubit Pauli group, can be found in two companion papers.
\end{abstract}

\maketitle

\textbf{Introduction.} Unlike the intrinsically geometric structure of state spaces in classical physics, in the form of symplectic/Poisson manifolds, quantum mechanics is first and foremost an algebraic framework. States are normalised, positive linear functionals on the noncommutative von Neumann or $C^*$-algebra, whose operators contain the observables of the theory. This, together with the richer structure of correlations described by quantum states - including nonlocal and contextual correlations - impedes an intuitive geometric perspective on the set of states describing a (possibly composite) quantum system.

In quantum information and computation, quantum systems are usually taken to be composed from $n$-fold tensor powers of a single unit: a `qubit', that is, a quantum system of Hilbert space dimension two. Most computational architectures are devised for $n$-qubit systems \footnote{Models for higher-dimensional `qudits' exist as well \cite{WangHuSandersKais2020} - and sometimes are preferrable to qubits - yet, their study is much less developped than their qubit counterparts.}, which motivates their study in particular. Yet, even in this special case, few structural results (beyond the case of two qubits) are known \cite{JakobczykSiennicki2001,BengtssonZyczkowski2006,UskovRau2008,MorelliEtAl2024,Mosseri_2001,Bernevig_2003,Sudbery_2001,Jing_2015}.

For a single qubit, every density operator is of the form $\rho=\frac{1}{2}\one+\frac{1}{2}(\alpha_X X+\alpha_Y Y+\alpha_Z Z)$, where $\alpha_X,\alpha_Y,\alpha_Z\in\R$ and $\aP:=\{\one,X,Y,Z\}$ is the set of Hermitian Pauli matrices,
\begin{align*}
    \one&=\begin{pmatrix}
        1 & 0 \\ 0 & 1
    \end{pmatrix} &
    X&=\begin{pmatrix}
        0 & 1 \\ 1 & 0
    \end{pmatrix} &
    Y&=\begin{pmatrix}
        0 & -i \\ i & 0
    \end{pmatrix} &
    Z&=\begin{pmatrix}
        1 & 0 \\ 0 & -1
    \end{pmatrix}\; .
\end{align*}
Pure states lie on the `Bloch sphere', described by, $\alpha^2_X+\alpha^2_Y+\alpha^2_Z=1$, whereas general mixed states fill the inner `Bloch ball', $\alpha^2_X+\alpha^2_Y+\alpha^2_Z<1$. Yet, this geometric picture quickly becomes complicated for $n\geq2$ qubits, since unlike for $n=1$, the (pure-state) space is much more constrained: there are many more purity constraints than the single one $\sum_{\one\neq P\in\aP_n}\alpha^2_P=2^n-1$. While some representations exist for two- and three-qubit systems \cite{BengtssonZyczkowski2006,Mosseri_2001,Bernevig_2003,HoehnWever2017}, the situation for $n$ qubits remains structurally opaque.

Another subject in which the characterization of quantum states assumes a central role is reconstructions of quantum theory from operational axioms \cite{Hardy2001,Masanes:2010tt,dakic2011quantum,chiribella2011informational,selby2021reconstructing,muller2021probabilistic,Hoehn2017,HoehnWever2017}. Specifically, the informational reconstruction in Refs.~\cite{Hoehn2017,HoehnWever2017} (see \cite{Hoehn:2016otu,Hohn:2017cpr} for summaries) leads to certain `pentagon identities' for two qubits. These are pure state complementarity equalities in the quadratic Brukner-Zeilinger information measure \cite{BruknerZeilinger1999,BruknerZeilinger2001} that characterise the set of pure states and the unitary group. To the best of our knowledge, this remains the only  `prediction' by any reconstruction of novel structural properties in the standard formalism of quantum theory. To complete the reconstruction  for arbitrarily many qubits in Ref.~\cite{Hoehn:2016otu}, certain universality results made it feasible to sidestep the question what the analogous informational complementarity identities are for $n>2$ that similarly characterize the space of pure states and the unitary group.

Here, we finally answer this question, deriving a family of remarkably simple informational invariants of $n$-qubit pure states in Bloch representation, while expounding on the complementarity properties of the associated observable sets and revealing previously unnoticed structure in the $n$-qubit Pauli group in Ref.~\cite{FrembsHoehn2026b}. Furthermore, Ref.~\cite{Frembs2026} proves that this family fully encodes the unitary orbit of $n$-qubit pure states, thus providing a new characterisation of the intricate geometry of the $n$-qubit state space, generalising this the two-qubit case in Ref.~\cite{HoehnWever2017}.

\section{Background and Preliminaries}\label{sec: background and preliminaries}

In this section, we fix our notation, introduce the Brukner-Zeilinger information measure \cite{BruknerZeilinger1999,BruknerZeilinger2001}, review the known two-qubit invariants from Ref.~\cite{HoehnWever2017}, and demonstrate that a straightforward generalisation fails.

\subsection{Hermitian $n$-qubit Pauli operators}\label{sec: Pauli operators}

Denote the \emph{$n$-qubit Pauli group} by $\cP_n:=\langle\one,X,Y,Z\rangle^{\otimes n}$, and its normaliser, the \emph{Clifford group}, by $\mathrm{Cl}_n=\mathrm{Aut}(\cP_n)$. The subset $\aP_n:=\{\one,X,Y,Z\}^{\otimes n}$ forms an orthogonal basis in the space of Hermitian matrices, with respect to the Hilbert-Schmidt inner product $(A,B):=\tr[A^*B]$. It follows that every density operator $\rho\in\mc{D}(\cH)$ for $\cH=\C^{2^n}$ inside the space of density operators $\mc{D}(\cH)$ has a \emph{Bloch representation $\rho=\frac{1}{2^n}(\one+\sum_{\one\neq P\in\aP_n}\alpha_P P)$} whose \emph{Bloch components $\alpha_P\in\R$} satisfy $|\alpha^2|=\sum_{\one\neq P\in\aP_n}\alpha^2_P\leq 2^n-1$, with equality for pure states.

As our main quantity of interest (see Eq.~(\ref{eq: BZ information measure})) is further invariant under re-phasing $P\ra (-1)^{v(P)}P$ for $v:\aP_n\ra\zz_2$ arbitrary, we will disregard the central phase group $K=\{\pm\one,\pm i\one\}<\cP_n$ and restrict to the Hermitian Pauli representatives $\aP_n\simeq\cP_n/K$ (as sets) instead (see App.~\ref{app: Abelian Pauli group}). In particular, when considering (anti-)commutation relations between elements in $\aP_n$, we will identify $\zeta P\in\cP_n$ for $\zeta\in K$ and $P\in\aP_n$ with its unique representative in $\aP_n$, and write $\zeta P\sim P$.

\subsection{Pentagon identities for two qubits}\label{sec: pentagon identities}

In Ref.~\cite{BruknerZeilinger1999,BruknerZeilinger2001}, Brukner and Zeilinger argue that the Shannon entropy, defined for classical random variables, fails to capture a crucial aspect of quantum mechanics: complementarity. Unlike in classical physics, observables in quantum theory can in general not be measured jointly with arbitrary precision. As a consequence, information about a quantum system is generally spread out across complementary observables. To capture this, they propose a quadratic measure in the Bloch components of a quantum state. This measure also emerges in the axiomatic reconstruction of qubit quantum theory in Refs.~\cite{Hoehn2017,HoehnWever2017}, which postulates the existence of a finite set of complementary observables - which, given the axioms, are shown to correspond with Hermitian $n$-qubit Pauli matrices $\aP_n$ - and quantifies the information contained in a given (generically complementary) subset $\PI\subset\aP_n$ in a given state \footnote{That states correspond with quantum states is, of course, part of the reconstruction in Ref.~\cite{Hoehn2017,HoehnWever2017}.} $\rho=\frac{1}{2^n}(\one+\sum_{\one\neq P\in\aP_n}\alpha_P P)$ to be 
\begin{align}\label{eq: BZ information measure}
    I_\rho(\PI)
    :=\sum_{Q\in\PI}\left(\tr[\rho Q]\right)^2
    =\sum_{Q\in\PI}\alpha^2_Q\; ,
\end{align}
which for pure states $\rho=\dyad{\psi}$ is equivalent to $I_\psi(\PI)=\sum_{Q\in\PI}|\langle\psi|Q|\psi\rangle|^2=\sum_{Q\in\PI}\tr[(\rho Q)^2]$. We call $\PI\subset\aP_n$ a \emph{purity invariant}  if $I_\psi(\PI)=\mathrm{const}$ for all pure states $\psi\in\CP^{2^n-1}$ \footnote{The full set $\aP=\{\one,X,Y,Z\}^{\otimes n}$ and the identity always yield (trivial) identities, $I_\tot:=\sum_{P\in\aP_n}I_\psi(P)-1=2^n-1$ and $I_\psi(\one)=1$. For $n=1$ this also yields a maximal set of mutually anti-commuting Pauli operators and this identity too is crucial for recovering the Bloch ball as the correct state space of a single qubit \cite{Hoehn2017}.}. A crucial ingredient of the reconstruction in Ref.~\cite{Hoehn2017,HoehnWever2017} is that for $n=2$, there exist non-trivial  such invariants in the form of maximal sets of mutually anti-commuting Pauli operators, where maximality refers to both unextendability \emph{and} cardinality \footnote{Maximal cardinality in our notion of maximality is necessary as not all unextendible sets of mutually anti-commuting Paulis have the same cardinality \cite{HoehnWever2017}.}. For example, for $\pent=\{X\one,Z\one,YX,YY,YZ\}$ one finds \footnote{For brevity, we sometimes omit tensor products, and write $PQ$ for $P\otimes Q$ with $P,Q\in\aP_n$. Overloading notation slightly, we will also write $\one$ for the identity in $M_{2^n}(\C)$ (the explicit dimension will be clear from context).}
\begin{align}\label{eq: pentagon identity}
   I_\psi(\pent)
   =\alpha_{X\one}^2+\alpha_{Z\one}^2+\alpha_{YX}^2+\alpha_{YY}^2+\alpha_{YZ}^2=1\; ,  
\end{align}
for every pure state $\psi\in\CP^3$. Since maximal sets of mutually anti-commuting two-qubit Pauli operators have cardinality $5=2n+1$ for $n=2$ \cite{SarkarVanDenBerg2021}, they are called `pentagon identities' in Ref.~\cite{HoehnWever2017}, and as shown in that reference they characterise the sets of pure states and the unitary group. Given the results in Ref.~\cite{Hoehn2017,HoehnWever2017}, it is natural to ask whether such (non-trivial) purity invariants exist also for $n>2$ qubits.

\subsection{A first (failed) generalisation attempt}\label{sec: naive generalisation}

Since maximal sets of anti-commuting Paulis give rise to purity invariants for $n=1,2$ one might suspect this to generalise to $n>2$. However, this is not the case.

\begin{proposition}\label{prop: naive generalisation}
    Maximal sets $\cS\subset\aP_n$ of mutually anti-commuting Paulis are not purity invariants for $n>2$, that is, $I_\psi(\cS)\neq I_{\psi'}(\cS)$ in general for pure states $\psi\neq\psi'$.
\end{proposition}

\begin{proof}
    We provide an explicit counterexample for $n=3$. Consider the maximally anti-commuting set of three-qubit Pauli operators \cite{SarkarVanDenBerg2021}, of the form
    \begin{align*}
        \cS
        =\{XXX, XYX, XZX, X\one Y, X\one Z, Y\one\one, Z\one\one\}\; ,
    \end{align*}
    and the pure state $\psi\in\CP^{2^3-1}$ in Bloch representation,
    \begin{align*}
        \rho=\dyad{\psi} 
        =\frac{1}{8} (&\one\one\one+XX\one+YY\one+ZZ\one\\
        &+XXY+YYY+ZZY+\one\one Y)\; .
    \end{align*}
    Note that $\psi$ is a pure state as the simultaneous eigenstate of a maximal set of mutually commuting Pauli operators. One computes
    ${I_\psi(\cS)
    =\sum_{P\in\cS} \tr[\dyad{\psi}P]^2
    =0}$, while e.g.\ $I_{\psi'}(\cS)=1$ for $|\psi'\rangle=|000\rangle$.
\end{proof}

It follows that the cases $n=1,2$ are exceptional (whether or not purity invariants exist for more than two qubits). In the next section, we will see why $n=1,2$ do in fact represent exceptional cases in a general family of purity invariants for $n$ qubits with $n\in\mathbb{N}$.

\section{Family of purity invariants}\label{sec: family of pure state invariants}

Prop.~\ref{prop: naive generalisation} implies that both the complementarity relations and the algebraic structure of mutually anti-commutating Pauli operators constituting the pentagon identities for two qubits require generalisation. We therefore reevaluate them from these two perspectives.

To extract their complementarity structure, consider the partition of $\aP_2\backslash\{\one\}$ into $5$ maximal Abelian subgroups (arranged vertically) with the identity removed:
\begin{center}
\setlength{\tabcolsep}{12pt}
\renewcommand{\arraystretch}{1.5}
\begin{tabular}{ccccc}
    $X\one$ & $Z\one$ & $YX$ & $YY$ & $YZ$ \\
    $\one X$ & $\one Z$ & $ZY$ & $Y\one$ & $XY$ \\
    $XX$ & $ZZ$ & $XZ$ & $\one Y$ & $ZX$
\end{tabular}
\end{center}
The operators in every column mutually commute and thus share a common eigenbasis; moreover, the eigenbases of different columns are easily seen to be mutually unbiased \cite{LawrenceBruknerZeilinger2002}. Note that the first row corresponds with $\pent$ above. Its complementarity structure can therefore be understood in terms of mutual unbiasedness. The latter reduces to mutual anti-commutation if one picks a single element in every maximal Abelian column subgroups in the partition, but is readily generalised to subgroups of higher dimension (with the identity included). 

In other words, maximal complementarity between individual observables becomes supplanted by maximal complementarity between \emph{sets} of compatible (commuting) observables; maximal information about one set of compatible observables comes at the expense of total ignorance about another. For a detailed study of such maximal complementarity structures in the Pauli group and their relation with mutual unbiasedness, see Ref.~\cite{FrembsHoehn2026b}. As we show there, maximal complementarity structures give rise to purity invariants, of which the family in Def.~\ref{def: family of purity invariants} below constitutes a special case with particularly nice algebraic closure properties (see Ref.~\cite{Frembs2026}).

To gain further insight into its algebraic structure, let us parametrise the elements in $\pent$ in terms of operators in the first two columns, that is, in terms of the Weyl representation (see App.~\ref{app: Abelian Pauli group}). In order to do so, in addition to the elements $X_1,Z_1\in\pent$ \footnote{Here, $X_i/Z_i$ denote the Pauli strings with only non-identity factor $X/Z$ in the $i$-th slot, e.g.\ $X_3=\one\one X\one\cdots\one$.}, we need to specify two further elements, one from the first and one from the second column. It turns out that a canonical choice is to choose an element $g^Z_X$ from the first column, which commutes with $Z_1$, and an element $g^X_Z$ from the second column, which commutes with $X_1$, since these constraints fix these elements uniquely to be $g^Z_X=X_2$ and $g^X_Z=Z_2$ (as well as more generally under the generalisation of the elements $X_1,Z_1\in\pent$ to complementary subgroups $\tA_X$ and $\tA_Z$, see below and Ref.~\cite{FrembsHoehn2026b}). With respect to these elements, the set $\pentone:=\pent\cup\{\one\}$ can be written in terms of the following (anti-)commutators:
\begin{equation}\label{eq: two-qubit relations}
\begin{aligned}
    \pentone
    =\{&\{\one,\one\},\{X\one,\one\},\{\one,Z\one\}, & &\{XX,ZZ\},\\
    &[X\one,ZZ], & &[XX,Z\one]\}\; .
\end{aligned}
\end{equation}
$\pentone$ thus admits a straightforward parametrisation with respect to elements $\tA_X:=\pentone\cap A_X=\langle X_1\rangle$ and $\tA_Z:=\pentone\cap A_Z=\langle Z_2\rangle$ in the intersection with the subgroups $A_X=\langle X_i\rangle_{i=1}^2$ and $A_Z=\langle Z_i\rangle_{i=1}^2$, together with two elements $A_X\ni g^Z_X\notin\pent$ and $A_Z\ni g^X_Z\notin\pent$, uniquely determined by $[g^Z_X,\tA_Z]=0=[g^X_Z,\tA_X]$. 

Moreover, this parametrisation applies to all pentagon identities in Ref.~\cite{HoehnWever2017}, and thus provides an alternative characterisation of their complementarity structure, e.g.\ the (anti-)commutation relations in Eq.~(\ref{eq: two-qubit relations}) encode that two non-identity elements in $\pentone$ anti-commute. Yet, unlike mutual anti-commutation, the above parametrisation generalises to arbitrarily many qubits as follows.

\begin{definition}\label{def: family of purity invariants}
    For $n\geq 1$, let $A_X=\langle X_i\rangle_{i=1}^n$, $A_Z=\langle Z_i\rangle_{i=1}^n$. For any subgroups $\tA_X<A_X$, $\tA_Z<A_Z$ of ($\zz_2$-)dimension $n-1$ and elements $g^Z_X\in A_X\backslash\tA_X$, $g^X_Z\in A_Z\backslash\tA_Z$ such that $[g^Z_X,\tA_Z]=0=[g^X_Z,\tA_X]$ (and thus $[g^Z_X,g^X_Z]\neq 0$), define
    \begin{widetext}
    \begin{align}\label{eq: family of purity invariants}
    \begin{split}
        \cJ(\tA_X,\tA_Z,g^Z_X,g^X_Z)\ 
        &:=\cJ_{\{\tA_X,\tA_Z\}}\ \cupdot\ \cJ_{\{g^Z_X\tA_X,g^X_Z\tA_Z\}}\ \cupdot\ \cJ_{[\tA_X,g^X_Z\tA_Z]}\ \cupdot\ \cJ_{[g^Z_X\tA_X,\tA_Z]}\\
        &:=\big\{\{\ta_X,\ta_Z\}\,\big\mid \,\ta_X\in\tA_X,\ta_Z\in\tA_Z\big\}\\
        &\hspace{1.95cm}\cupdot
        \big\{\{g^Z_X\ta_X,g^X_Z\ta_Z\}\,\big\mid \,\ta_X\in\tA_X,\ta_Z\in\tA_Z\big\}\\
        &\hspace{4.7cm}\cupdot \big\{[\ta_X,g^X_Z\ta_Z]\,\big\mid \,\ta_X\in\tA_X,\ta_Z\in\tA_Z
        \big\}\\
        &\hspace{7.05cm}\cupdot
        \big\{[g^Z_X\ta_X,\ta_Z]\,\big\mid \,\ta_X\in\tA_X,\ta_Z\in\tA_Z\big\}\; ,
    \end{split}
    \end{align}
    \end{widetext}
    where $[\cdot,\cdot]$ ($\{\cdot,\cdot\}$) denote the (normalised) (anti-) commutator, respectively (cf.~App.~\ref{app: Abelian Pauli group}).
\end{definition}

In App.~\ref{app: n=3}, we list the sets in Def.~\ref{def: family of purity invariants}, with their defining data highlighted in colour for the three-qubit case $n=3$, as the first non-trivial generalisation beyond two qubits.

From Def.~\ref{def: family of purity invariants}, $\tA_X\cupdot g_X^Z\tA_X = \langle\tA_X, g_X^Z\rangle=A_X$ (similarly for $X\leftrightarrow Z$). A Venn diagram depicting the set relations of the group elements defining $J$ is shown in Fig.~\ref{fig:venn_diagram}

The sets in Eq.~(\ref{eq: family of purity invariants}) are indeed disjoint, as they contain products ($\{Q,P\}\sim QP$ if $\{Q,P\}\neq 0$ and $[Q,P]\sim QP$ if $[Q,P]\neq 0$) involving different combinations of $g^Z_X$ and $g^X_Z$ \footnote{This also follows from Lm.~\ref{lm: conjugate generators} since disjointness is not affected by conjugation}. Their respective cardinalities are given in Lm.~\ref{lm: cardinalities} of App.~\ref{app: Abelian Pauli group}; in particular, $|\cJ|=2^{n-1}(2^n-1)$.

We write $\fg=\aP_n\backslash\cJ$ for the complement of $\cJ$ in $\aP_n$. Since $I_\psi(\aP_n)=2^n$ it follows immediately that $I_\psi(J)=\mathrm{const}$ if and only if $I_\psi(L)=\mathrm{const}$. Note that for $n=1$, Def.~\ref{def: family of purity invariants} gives $\cJ=\{\one\}$ and $\fg=\aP^1\backslash\{\one\}=\{X,Y,Z\}$, which corresponds with the trivial invariant $I_\psi(\fg)=1$ for pure states $\psi\in\CP^1$ \footnote{There are no other (nontrivial) purity invariants in this case, since a single-qubit pure state $\psi\in\mathbb{PC}^1$ is described by two real parameters, which coincides (only for $n=1$) with the parametrisation of the `Bloch sphere', defined by $\tr[\rho^2]=1(=\tr[\rho])$, that is, $\sum_{\one\neq P\in\aP_n}\alpha^2_P=1$ (and $\alpha_\one=\frac{1}{2}$).}. Moreover, Def.~\ref{def: family of purity invariants} recovers the pentagon identities from Ref.~\cite{HoehnWever2017} (see above).

Before turning to the proof that the sets $J$ define purity invariants for all $n$, we first show that the choice of reference subgroups in Def.~\ref{def: family of purity invariants} is without loss of generality.

\begin{figure}[htbp]
    \centering
    \includegraphics[width=0.3\textwidth]{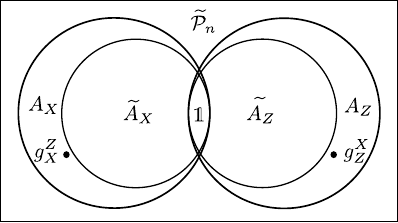} 
    \caption{Venn diagram showing the set structure of the relevant group elements defining the  sets in Def.~\ref{def: family of purity invariants}.}
    \label{fig:venn_diagram}
\end{figure}

\begin{lemma}\label{lm: conjugate generators}
    Let $A,B<\aP_n$ be two maximal Abelian subgroups such that $A\cap B=\{\one\}$. Moreover, let $\tA<A$, $\tB<B$ be Abelian subgroups of dimension $n-1$ and let $a\in A$, $b\in B$ satisfy the conditions in Def.~\ref{def: family of purity invariants}, that is, $\langle\tA,a\rangle=A$, $\langle\tB,b\rangle=B$ and $[\tA,b]=0=[\tB,a]$. Then there exists a Clifford unitary $U\in\Cl_n$ such that
    \begin{align*}
        U\tA U^\dagger&=\tA_Z=\langle Z_i\rangle_{i=1}^{n-1} &
        UaU^\dagger&=Z_n \\ 
        U\tB U^\dagger&=\tA_X=\langle X_i\rangle_{i=1}^{n-1} &
        UbU^\dagger&=X_n\; .
    \end{align*}
\end{lemma}

\begin{proof}
    We give the proof in App.~\ref{app: Abelian Pauli group}.
\end{proof}

Clifford conjugation thus yields a large class of purity invariant candidate sets, which satisfy the same algebraic properties as in Def.~\ref{def: family of purity invariants}. Whenever we write $J$ in what follows, we refer to any such set in this larger class.\\

\textbf{Algebraic properties.} A key property of the partition $\aP_n=\cJ\cupdot\fg$, given by Def.~\ref{def: family of purity invariants}, is that it satisfies the following closure relations with respect to (anti-)commutators between (elements in) $\cJ$ and $\fg$.

\begin{lemma}\label{lm: algebraic closure}
    $\cJ$ is closed under anti-commutators and $\fg$ is closed under commutators. Moreover, for $n\geq 2$,
    \vspace{-.2cm}
    \begin{equation}\label{eq: algebraic closure relations}
    \begin{aligned}
        \{\cJ,\cJ\}&=\cJ &\ \
        \{\cJ,\fg\}&=\fg &\ \
        \{\fg,\fg\}&=\cJ \\ 
        [\cJ,\cJ]&=\fg &\ \ 
        [\cJ,\fg]&=\cJ &\ \ 
        [\fg,\fg]&=\fg\\
    \end{aligned}
    \end{equation}
    whereas $[\cJ,\cJ]=[\cJ,\fg]=0$ for $n=1$.
\end{lemma}

\begin{proof}[Proof sketch]
    Recall that $\{Q,P\}\sim QP$ if $\{Q,P\}\neq 0$, similarly $[Q,P]\sim QP$ if $[Q,P]\neq 0$. To prove the first identity $\{\cJ,\cJ\}=\cJ$, we show that a non-vanishing anti-commutator of commuting elements in the various subsets of $\cJ$ in Eq.~(\ref{eq: family of purity invariants}) is of the form in Eq.~(\ref{eq: family of purity invariants}).
    
    Consider first $0\neq\{\ta_X,\ta_Z\},\{\ta'_X,\ta'_Z\}\in\cJ$ with $\{\{\ta_X,\ta_Z\},\{\ta'_X,\ta'_Z\}\}\neq0$. Hence, $[\ta_X,\ta_Z]=0=[\ta'_X,\ta'_Z]$, and $[\{\ta_X,\ta_Z\},\{\ta'_X,\ta'_Z\}]=0$. By Lm.~\ref{lm: nested commutativity}, the latter is equivalent to $[\ta_X,\ta'_Z]=0$ if and only if $[\ta'_X,\ta_Z]=0$, which implies $[\ta_X\ta'_X,\ta_Z\ta'_Z]=0$. Thus,
    \begin{widetext}
    \begin{align*}
        \{\{\ta_X,\ta_Z\},\{\ta'_X,\ta'_Z\}\}
        \sim\{\ta_X\ta_Z,\ta'_X\ta'_Z\}
        \sim
        \ta_X\ta_Z\ta'_X\ta'_Z
        \sim
        \ta_X\ta'_X\ta_Z\ta'_Z
        \sim\{\ta_X\ta'_X,\ta_Z\ta'_Z\}\in\cJ\; ,
    \end{align*}
    \end{widetext}
    where the last step uses the definition of $\cJ_{\{\tA_X,\tA_Z\}}$ in Eq.~(\ref{eq: family of purity invariants}). The remaining cases are analysed analogously; for completeness, we list them in App.~\ref{app: algebraic closure}.
\end{proof}

Lm.~\ref{lm: algebraic closure} shows that (anti-)commutation relations are (globally) encoded in terms of membership in $\cJ,\fg$, respectively. However, as we explicate in detail in Ref.~\cite{Frembs2026}, not all subsets $\PI\subset\aP_n$ which satisfy closure relations as in Eq.~(\ref{eq: algebraic closure relations}) give rise to purity invariants with respect to $I$.

In order to prove that Def.~\ref{def: family of purity invariants} yields purity invariants, we need to identify an additional property satisfied by these sets. To this end, note first that $\cJ$ contains Abelian subgroups of dimension (up to) $n-1$, for instance, $\tA=\langle\ta_X,\tA'_Z(\ta_X)\rangle=\{\{\ta_X,\ta_Z\}\mid\ta_Z\in\tA'_Z(\ta_X)\}\subset\cJ_{\{\tA_X,\tA_Z\}}\subset\cJ$, where $\tA'_Z(\ta_X):=C(\ta_X)\cap\tA_Z$ is the largest subgroup in $\tA_Z$ that commutes with some $\ta_X\in\tA_X$. That $\tA'_Z(\ta_X)$ is indeed a subgroup of dimension $n-2$ is easily seen from the canonical form of Def.~\ref{def: family of purity invariants} given by Lm.~\ref{lm: conjugate generators}. Conversely, 

\begin{lemma}\label{lm: no max ab subgroup}
    $\cJ$ contains no maximal Abelian subgroup.
\end{lemma}

\begin{proof}
    We give the proof in App.~\ref{app: no max ab subgroup}.
\end{proof}

It follows from Lm.~\ref{lm: algebraic closure} and Lm.~\ref{lm: no max ab subgroup} that $I_\psi(\cJ)\leq2^{n-1}$ for all stabiliser states $|\psi\rangle$ (cf.~Ref.~\cite{FrembsHoehn2026b}), that is, for all pure states whose Bloch representation has support on a maximal Abelian subgroup $A$ only, with corresponding Bloch components $\alpha_{P\in A}=1$. To prove equality for all pure states, we also need the following relation between Bloch components of $n$-qubit pure states. To state it, we define by ${C(P):=\{Q\in\aP_n\mid [P,Q]=0\}}$ the \emph{commutant of $P$ (in $\aP_n$)}, and analogously by ${C_\PI(P):=C(P)\cap\PI}$ the commutant of $P$ within the subset $\PI\subset\aP_n$.

\begin{lemma}\label{lm: information in general Pauli subsets}
    Let $\PI\subset\aP_n$  and $\rho=\frac{1}{2^n}(\one+\sum_{\one\neq P\in\aP_n}\alpha_P P)\in\CP^{2^n-1}$ be an $n$-qubit pure state. Then
    \begin{align}\label{eq: information in general Pauli subsets}
        I_\rho(\PI)
        =\frac{1-2^{n-1}}{2^{n-1}} |\PI|+\frac{1}{2^{n-1}}\sum_{\one\neq P\in\aP_n} |C_\PI(P)|\alpha^2_P\; .
    \end{align}
\end{lemma}

\begin{proof}
    We give the proof in App.~\ref{app: proof Pauli identity}.
\end{proof}

Lm.~\ref{lm: information in general Pauli subsets} suggests that purity invariants, $I_\psi(\PI)=\mathrm{const}$ for all $\psi\in\CP^{2^n-1}$, arise from sets $\PI\subset\aP_n$ whose commutants have cardinality independent of $P$. This is indeed the case, and holds for the sets defined in Def.~\ref{def: family of purity invariants}.

\begin{lemma}\label{lm: constant commutants}
    The cardinalities of the commutants ${C_\cJ(P\in\cJ)}$ and $C_\fg(P\in\cJ)$ are independent of $\one\neq P\in\cJ$, while the cardinality of the commutants $C_\cJ(P\in\fg)$ and $C_\fg(P\in\fg)$ are independent of $P\in\fg$.
\end{lemma}

\begin{proof}
    We give the proof in App.~\ref{app: constant commutants}.
\end{proof}

\textbf{Main result.} Combining the above lemmata shows that Def.~\ref{def: family of purity invariants} indeed defines purity invariants.

\begin{theorem}\label{thm: invariance}
    For every $n\geq 1$, subset $\cJ\subset\aP_n$ in Def.~\ref{def: family of purity invariants} (extended via Clifford conjugation in Lm.~\ref{lm: conjugate generators}) and pure state $\psi\in\CP^{2^n-1}$, it holds $I_\psi(\cJ)=2^{n-1}$.
\end{theorem}

\begin{proof}
    It follows with Lm.~\ref{lm: information in general Pauli subsets} and Lm.~\ref{lm: constant commutants} that
    \begin{widetext}
    \begin{align*}
        I_\psi(\cJ)
        &=\frac{1-2^{n-1}}{2^{n-1}}|\cJ|+\frac{1}{2^{n-1}}\sum_{\one\neq P\in\aP_n} |C_\cJ(P)|\alpha^2_P\\
        &=\frac{1-2^{n-1}}{2^{n-1}}|\cJ|+\frac{1}{2^{n-1}}|C_\cJ(\one\neq P\in\cJ)|I_\psi(\cJ)+\frac{1}{2^{n-1}}|C_\cJ(P\in\fg)|(2^n-1-I_\psi(\cJ))\; ,
    \end{align*}
    \end{widetext}
    where we used that $\sum_{\one\neq P\in\aP_n}\alpha^2_P=2^n-1$ for the Bloch components of pure states. Consequently, $I_\psi(\cJ)$ is constant for all $\psi\in\CP^{2^n-1}$, hence, equals the value for stabiliser states, which reads $I_\psi(\cJ)=2^{n-1}$, by Lm.~\ref{lm: no max ab subgroup}.
\end{proof}

We show in Ref.~\cite{FrembsHoehn2026b} that \emph{all} maximal complementarity sets, of which the $J$-sets are a subclass, constitute purity invariants. Furthermore, in Ref.~\cite{Frembs2026} we prove that the family of purity invariants in Def.~\ref{def: family of purity invariants} exhausts all invariants in the Brukner-Zeilinger information measure of Eq.~(\ref{eq: BZ information measure}) that are closed under products of commuting Paulis - and that, without this condition, many more purity invariants arise that have no analogue to the two-qubit case (and that may not correspond to complementarity sets). Despite this proliferation, we show that the purity invariants in Def.~\ref{def: family of purity invariants} fully encode the symmetry group $\mathrm{PSU}(2^n)\rtimes\zz_2<\mathrm{SO}(\aP_n)$ of the pure-state manifold $\CP^{2^n-1}$, thereby generalizing this observation from the $n=2$ case in \cite{HoehnWever2017} to arbitrary $n$.\\

\textbf{Counting invariants.} Given a purity invariant $I_\psi(\PI)=\mathrm{const}$ for all $\psi\in\CP^{2^n-1}$ with $\PI\subset\aP_n$, any Clifford conjugated set $C\PI C^{-1}$ defines another such set. We now show that all sets in Def.~\ref{def: family of purity invariants} arise in this way.

\begin{lemma}\label{lm: Clifford covariance}
    The family of purity invariants in Def.~\ref{def: family of purity invariants} defines an orbit under conjugation by Clifford unitaries.
\end{lemma}

\begin{proof}
    We give the proof in App.~\ref{app: Clifford invariance}.
\end{proof}

We close with an explicit count of the number of different invariants in Def.~\ref{def: family of purity invariants}.

\begin{theorem}\label{thm: counting}
    There exist $2^{n-1}(2^n-1)$ distinct Clifford-conjugate subsets $\cJ$ in $\aP_n$.
\end{theorem}

\begin{proof}
    We give the proof in App.~\ref{app: counting}
\end{proof}

For $n=2$, we recover the six pentagon sets of \cite{HoehnWever2017}. Similarly, in App.~\ref{app: n=3} we exhibit the 28 $J$-sets for $n=3$.
\vspace{.1cm}

\section{Conclusion and Outlook}\label{sec: outlook}

We presented a family of quadratic informational purity invariants for $n$-qubit pure states. This generalises known invariants in the two-qubit case \cite{Hoehn2017,HoehnWever2017}, which we find to be exceptional compared to the rich algebraic structure that emerges for $n>2$ qubits. We provided an explicit representation of these invariants, and highlighted some of their properties. A more detailed study, as well as a full characterisation of all purity invariants with respect to the measure in Eq.~(\ref{eq: BZ information measure}) will appear in \cite{Frembs2026}.

We finish by mentioning a few potential applications. First, the structure of purity invariants obtains an interpretation in terms of maximal complementarity sets, which give rise to strong uncertainty relations that in turn are key in cryptographic protocols. Furthermore, the invariants $I_\psi(\mathcal{S})=const$ in the Brukner-Zeilinger measure admit the interpretation of \emph{informational complementarity equalities for $n$-qubit pure states}. We discuss this aspect in detail in Ref.~\cite{FrembsHoehn2026b}. For these reasons, we expect our family of purity invariants to find use in benchmarking tomography, quantum error correction \cite{Fujii2025}, as well as the study of uncertainty relations. We collect some first supporting evidence along those lines in Ref.~\cite{Frembs2026}.

\textbf{Acknowledgments.} MF thanks Lukas Hantzko, Paul Herringer and Robert Raussendorf for discussions, and Lukas Hantzko for proof-reading of an earlier draft. This work was made possible through the support of the ID\# 62312 grant from the John Templeton Foundation, as part of the project \href{https://www.qiss.fr/}{``The Quantum Information Structure of Spacetime'' (QISS)} and through the \href{https://withoutspacetime.org}{``WithOut SpaceTime project'' (WOST)}, led by the Center for Spacetime and the Quantum (CSTQ), and funded by Grant ID\# 63683 from the John Templeton Foundation. The opinions expressed in this work are those of the authors and do not necessarily reflect the views of the John Templeton Foundation.

\bibliography{bibliography}

\newpage
\appendix
\onecolumngrid

\section{Symplectic structure of Pauli group and proof of Lm.~\ref{lm: conjugate generators}}\label{app: Abelian Pauli group}

Let $\cP_n=\langle\one,X,Y,Z\rangle^{\otimes n}$ denote the $n$-qubit Pauli group. Let $K=\{\pm i,\pm 1\}=Z(\cP_n)$ denote its center, and let $\cP_n/K$ be the \emph{projective Pauli group}. Note that $\cP_n/K\cong\zz^{2n}_2$. In order to keep track of (anti-)commutation relations in $\cP_n/K$, we say that $\tP,\tQ\in\cP_n/K$ \emph{(anti-)commute} whenever any chosen representatives $P,Q\in\cP_n$ (anti)commute. We may thus define the (anti-)commutator in $\cP_n/K$ by writing $[\tP,\tQ]=0$ and $\{\tP,\tQ\}=\widetilde{PQ}$ whenever $[P,Q]=0$, and $[\tP,\tQ]=\widetilde{PQ}$ and $\{\tP,\tQ\}=0$ whenever $[P,Q]\neq 0$; we will write $(\cP_n/K,[\cdot,\cdot])$ for the projective Pauli group together with the commutator map on it. A subgroup $A<\cP_n/K$ is called \emph{Abelian} if its elements pairwise commute. By explicitly identifying (the elements of) $\cP_n/K$ with Hermitian representatives - e.g.\ with $\aP_n:=\{\one,X,Y,Z\}^{\otimes n}$ - in the main text, we will (for ease of readability) simply write $\zeta P\sim P$ for the (set) identification of $\cP_n/K$ with $\aP_n$, specifically when taking (anti-)commutators.  We will also write $A<\aP_n$ for Abelian subgroups $A\leq\cP_n/K$, as well as $C(\tQ)=C_{\aP_n}(\tQ):=\{\tP\in\aP_n\mid[\tQ,\tP]=0\}$ for the \emph{commutant of $\tQ$ in $\aP_n$}, extended to subsets in the obvious manner.

We compare this with the phase convention for Hermitian Pauli operators in the Weyl representation \cite{deBeaudrap2013}, 
\begin{align}\label{eq: Weyl operators}
    V=\zz^n_2\times\zz^n_2\ni v=(x,z)\mapsto W(v):=i^{x\cdot z}\otimes_{k=1}^n X^{x_k}Z^{z_k}\in\cP_n\; ,
\end{align}
where $x\cdot z=\sum_{k=1}^nx_kz_k$. The commutator in $\cP_n$ is determined by the symplectic form $\omega(v,w)=v^T\sigma w\mod 2$ for $\sigma=\begin{pmatrix} 0 & \one \\-\one & 0 \end{pmatrix}$, namely $W(v)W(w)=(-1)^{\omega(v,w)}W(w)W(v)$ (and vice versa). In particular, $[W(v),W(w)]=0$ if and only if $\omega(v,w)=0$ for all $v,w\in V=\zz^{2n}_2$. It follows that $(\cP_n/K,[\cdot,\cdot])\cong(V,\omega)$ for $V=\zz^{2n}_2$, and that Eq.~(\ref{eq: Weyl operators}) defines a Hermitian section $V\ra\pm\aP_n$ (of the central extension $1\ra\zz_4\ra\cP_n\ra V\ra 1$). Our results below apply to any choice of Hermitian section, that is, independently of any phase convention $P\ra (-1)^{v(P)}P$ for $v:\aP_n\ra\zz_2$. To emphasise this point, we will write Hermitian Pauli representatives (up to phase) using the phase-free Hermitian Pauli strings $\aP_n$.

Under this convention, the commutativity structure on elements $\aP_n$ is given as follows.

\begin{lemma}\label{lm: nested commutativity}
    Let $P_i,Q_j\in\aP_n$ with $P_i\sim W(v_i),Q_j\sim W(w_j)$ for $v_i,w_j\in\zz^{2n}_2$ and $i\in[m]$, $j\in[n]$, then
    \begin{align}
        [P_1\cdots P_m,Q_1\cdots Q_n]=0\ \Longleftrightarrow\ 
        \sum_{i,j=1}^{m,n}\omega(v_i,w_j)=0\; .
    \end{align}
\end{lemma}

\begin{proof}
    Since two $n$-qubit Pauli operators either commute or anti-commute, this follows by induction from
    \begin{equation*}
        [W(v_1)W(v_2),W(w)]=0\ \Longleftrightarrow\
        0
        =\omega(v_1+v_2,w)
        =\omega(v_1,w)+\omega(v_2,w)\; ,
    \end{equation*}
    where we use that $W(v)W(w)\sim W(v+w)$.
\end{proof}

In particular, note that every element $\one\neq P\in\aP_n$ (anti-)commutes with half of $\aP_n$, that is, $|C(P)|=\frac{1}{2}4^n$.\\

\textbf{Symplectic structure and Clifford group.} The \emph{$n$-qubit Clifford group $\Cl_n$} is the normaliser of the $n$-qubit Pauli group $\cP_n$, that is, $\Cl_n:=\langle U\in\cU(d)\mid UPU^\dagger\in\cP_n\ \forall P\in\cP_n\rangle$. Clearly, $\Cl_n$ contains the Pauli group $\cP_n$, and is otherwise generated, first, by the single-qubit Clifford operators (for every one of the $n$ qubits),
\begin{align}\label{eq: conjugating elements}
    S&=\begin{pmatrix}
        1 & 0 \\ 0 & i
    \end{pmatrix} & 
    H&=\frac{1}{\sqrt{2}} \begin{pmatrix}
        1 & 1 \\ 1 & -1
    \end{pmatrix}\; ,
\end{align}
which, in quantum computing parlance, are known as the phase and Hadamard gate, respectively; and, second by CNOT gates $\Lambda\!\mathrm{X}_{ij}$ between any pair of qubits $1\leq i<j\leq n$. Their action by conjugation on Pauli operators is given by
\begin{equation}\label{eq: CNOT}
\begin{aligned}
    Z_i\otimes\one_j&\stackrel{\Lambda\!\mathrm{X}_{ij}}{\longmapsto} Z_i\otimes\one_j &\hspace{3.5cm}
    X_i\otimes\one_j&\stackrel{\Lambda\!\mathrm{X}_{ij}}{\longmapsto} X_i\otimes X_j \\
    \one_i\otimes Z_j&\stackrel{\Lambda\!\mathrm{X}_{ij}}{\longmapsto} Z_i\otimes Z_j &\hspace{3.5cm}
    \one\otimes X_j&\stackrel{\Lambda\!\mathrm{X}_{ij}}{\longmapsto} \one\otimes X_j
\end{aligned}\; .
\end{equation}
The $n$-qubit Clifford group $\Cl_n$ contains the symmetric group $S_n$, generated by SWAP gates $V_{ij}$ for $1\leq i<j\leq n$, which decompose into CNOT gates as $V_{ij}=\Lambda\!\mathrm{X}_{ij}\Lambda\!\mathrm{X}_{ji}\Lambda\!\mathrm{X}_{ij}$, and whose action on Pauli operators is, of course,
\begin{align}\label{eq: SWAP}
    Z_i\otimes\one_j&\stackrel{V_{ij}}{\longmapsto}\one_i\otimes Z_j &
    X_i\otimes\one_j&\stackrel{V_{ij}}{\longmapsto}\one_i\otimes X_j\; .
\end{align}

We recall the following standard result about the symplectic structure of the $n$-qubit Pauli group \cite{NebeRainsSloane2001,DehaeneDeMoor2003,deBeaudrap2013}.

\begin{theorem}\label{thm: Pauli automorphism group}
    $\Sp(2n,\zz_2)=\Cl_n/(U(1)\cP_n)$.
\end{theorem}

Using this fact, the next lemmata provide a canonical representation for the subgroups in Def.~\ref{def: family of purity invariants}.

\begin{lemma}\label{lm: conjugate pairs}
    Let $A,B<\aP_n$ be two maximal Abelian subgroups such that $A\cap B=\{\one\}$. Then there exists a Clifford unitary $U\in\Cl_n$ such that
    \begin{align*}
        UAU^\dagger&=A_Z=\langle Z_i\rangle_{i=1}^n &
        UBU^\dagger&=A_X=\langle X_i\rangle_{i=1}^n\; .
    \end{align*}
\end{lemma}

\begin{proof}
    $A\cap B=\{\one\}$, implies that we can find generating sets $\cG=\{W(v_i)\}_{i=1}^n$ for $A=\langle\cG\rangle$ and $\cH=\{W(w_i)\}_{i=1}^n$ for $B=\langle\cH\rangle$ such that $\omega(v_i,w_j)=\delta_{ij}$, that is, such that $(v_1,\cdots,v_n,w_1,\cdots,w_n)$ is a canonical basis of the symplectic vector space $(V=\zz^{2n}_2,\omega)$. Indeed, (after projectivisation) $A,B$ correspond to Lagrangian subspaces $L_A,L_B$ in the symplectic vector space $(V=\zz^{2n}_2,\omega)$ with $L_A\cap L_B=0$ and thus $L_A\oplus L_B=V$. (Recall that an isotropic subspace is one on which the symplectic form, here induced by the commutator, vanishes; Lagrangian subspaces are maximal isotropic subspaces, and maximal here follows since $\aP_n$ contains no Abelian subgroups of dimension greater than $n$.) Hence, there exists a symplectic transformation $S$, which results in a change of basis $Sv_i=z_i$ and $Sw_i=x_i$, where $z_i=(\underbrace{0,\cdots,0}_{i-1},1,\underbrace{0,\cdots,0}_{2n-i-1})$ and $x_i=(\underbrace{0,\cdots,0}_{n+i-1},1,\underbrace{0,\cdots,0}_{n-i-1})$. The result thus follows since $\Sp(2n,\zz_2)<\Cl_n$ by Thm.~\ref{thm: Pauli automorphism group}.
\end{proof}

What is more, we have the following refinement of Lm.~\ref{lm: conjugate pairs} to the family of purity invariants in Def.~\ref{def: family of purity invariants}.

\setcounter{lemma}{1}

\begin{lemma}
    Let $A,B<\aP_n$ be two maximal Abelian subgroups such that $A\cap B=\{\one\}$. Moreover, let $\tA<A$, $\tB<B$ be Abelian subgroups of dimension $n-1$ and let $a\in A$, $b\in B$ satisfy the conditions in Def.~\ref{def: family of purity invariants}, that is, $\langle\tA,a\rangle=A$, $\langle\tB,b\rangle=B$ and $[\tA,b]=0=[\tB,a]$. Then there exists a Clifford unitary $U\in\Cl_n$ such that
    \begin{align*}
        U\tA U^\dagger&=\tA_Z=\langle Z_i\rangle_{i=1}^{n-1} &
        UaU^\dagger&=Z_n \\ 
        U\tB U^\dagger&=\tA_X=\langle X_i\rangle_{i=1}^{n-1} &
        UbU^\dagger&=X_n\; .
    \end{align*}
\end{lemma}

\begin{proof}
    By Lm.~\ref{lm: conjugate pairs}, we may assume that $A=A_Z=\langle Z_i\rangle_{i=1}^n$ and $B=A_X=\langle X_i\rangle_{i=1}^n$. By applying a combination of CNOT (incl SWAP) gates, whose action on Paulis is given by Eq.~(\ref{eq: CNOT}) and Eq.~(\ref{eq: SWAP}), it is not hard to see that we can map $a\mapsto Z_n$, $\tA\mapsto\tA_Z=\langle Z_1,\cdots,Z_{n-1}\rangle$, while keeping $A_X$ fixed. Moreover, since unitary conjugation preserves commutation relations, evaluating the commutation relations satisfied by $\tA,\tB$ and $a,b$ (those in Def.~\ref{def: family of purity invariants}) implies, first, $[Z_n,U\tB U^\dagger]=0$, hence, $\tA_X=\langle X_1,\cdots,X_{n-1}\rangle$, and second, $[UbU,Z_n]\neq 0$ and $[\tA_Z,UbU]=0$, and together $UbU=X_n$.
\end{proof}

\setcounter{lemma}{8}

\begin{lemma}\label{lm: cardinalities}
    The cardinalities of the disjoint subsets appearing in Eq.~(\ref{eq: family of purity invariants}) of Def.~\ref{def: family of purity invariants} are
    \begin{equation*}
    \begin{aligned}
        |\cJ_{\{\tA_X,\tA_Z\}}|
        &=(2^{n-1}+1)2^{n-2}\; , \\
        |\cJ_{\{g^Z_X\tA_X,g^X_Z\tA_Z\}}|
        &=|\cJ_{[\tA_X,g^X_Z\tA_Z]}|=|\cJ_{[g^Z_X\tA_X,\tA_Z]}|
        =(2^{n-1}-1)2^{n-2}\;.
    \end{aligned}
    \end{equation*}
    Consequently, $|\cJ|=(2^n+1)(2^{n-1}-1)+1=2^{n-1}(2^n-1)$.
\end{lemma}

\begin{proof}
    By Lm.~\ref{lm: conjugate generators}, we may consider w.l.o.g. the canonical case where $\tA_X=\langle X_i\rangle_{i=1}^{n-1}, g^Z_X=X_n, \tA_Z=\langle Z_i\rangle_{i=1}^{n-1}, g^X_Z=Z_n$ since the cardinalities are not changed by Clifford conjugating $\tA_X, g^Z_X, \tA_Z, g^X_Z$. Consider now any $\widetilde{a}_X=\otimes_{k=1}^{n-1} X^{x_k}\otimes\one\in\tA_X,\widetilde{a}_Z=\otimes_{k=1}^{n-1} Z^{z_k}\otimes\one\in\tA_Z$. The elements $\widetilde{a}_X$ and $\widetilde{a}_Z$ will commute if and only if $x\cdot z\mod 2=0$. If they commute, their product $\{\widetilde{a}_X, \widetilde{a}_Z\}=(\otimes_{k=1}^{n-1} X^{x_k}Z^{z_k})\otimes\one$ is made up of an even number of $Y$-Paulis, distributed over the first $n-1$ qubits, namely at the indices where both $x$ and $z$ are simultaneously non-zero, while the remaining first $n-1$ qubits will have either $\one, X$ or $Z$. Therefore, summing over the even number $k$ times that $Y$-Paulis can appear, one gets
    \begin{equation*}
        |\cJ_{\{\tA_X,\tA_Z\}}| = \sum_{\text{even}\, k=0}^{n-1} \binom{n-1}{k}\, 3^{n-1-k}=(2^{n-1}+1)2^{n-2}.
    \end{equation*}
    The other three cardinalities follow by a similar reasoning. Namely, consider $g^Z_X\widetilde{a}_X=(\otimes_{k=1}^{n-1} X^{x_k})\otimes X$ and $g^X_Z\widetilde{a}_Z=(\otimes_{k=1}^{n-1} Z^{z_k})\otimes Z$. The products $\{g^Z_X\widetilde{a}_X, g^X_Z\widetilde{a}_Z\}=(\otimes_{k=1}^{n-1} X^{x_k}Z^{z_k})\otimes Y$, $[\widetilde{a}_X, g^X_Z\widetilde{a}_Z]=(\otimes_{k=1}^{n-1} X^{x_k}Z^{z_k})\otimes Z$ and $[g^Z_X\widetilde{a}_X, \widetilde{a}_Z]=(\otimes_{k=1}^{n-1} X^{x_k}Z^{z_k})\otimes X$ are non-vanishing if and only if $x\cdot z\mod 2=1$, in which case they are made up of an odd number of $Y$-Paulis distributed over the first $n-1$ qubits, while the remaining first $n-1$ qubits will have either $\one, X$ or $Z$. Summing now over the odd number $k$ times that $Y$-Paulis can appear, one gets
    \begin{equation*}
        |\cJ_{\{g^Z_X\tA_X,g^X_Z\tA_Z\}}|=|\cJ_{[\tA_X,g^X_Z\tA_Z]}|=|\cJ_{[g^Z_X\tA_X,\tA_Z]}| = \sum_{\text{uneven}\, k=0}^{n-1} \binom{n-1}{k}\, 3^{n-1-k}=(2^{n-1}-1)2^{n-2}\; .\qedhere
    \end{equation*}
\end{proof}

\section{Proof of Lm.~\ref{lm: information in general Pauli subsets}}\label{app: proof Pauli identity}

\setcounter{lemma}{0}

\begin{lemma}
    Let $\PI\subset\aP_n$ and let $\rho=\dyad{\psi}=\frac{1}{2^n}(\one+\sum_{\one\neq P\in\aP_n}\alpha_P P)$ be a pure state $\psi\in\CP^{2^n-1}$. Then
    \begin{align*}
        I_\psi(\PI)
        =\frac{1-2^{n-1}}{2^{n-1}} |\PI|+\frac{1}{2^{n-1}}\sum_{\one\neq P\in\aP_n} |C_\PI(P)|\alpha^2_P\; .
    \end{align*}
\end{lemma}

\begin{proof}
    Recall that the Hermitian Paulis $\aP_n$ form an orthonormal basis in the space of Hermitian operators $H_{2^n}(\C)$, in particular, $\tr[PQ]=2^n\delta_{PQ}$ since $P^2=\one$ and $\tr[P]=0$ for all $\one\neq P,Q\in\aP_n$. Using this fact, together with $\tr[Q\dyad{\psi}]^2=|\langle\psi|Q|\psi\rangle|^2=\tr[(Q\dyad{\psi})^2]$ for pure states $\psi\in\CP^{2^n-1}$, we compute the information in $\PI\subset\aP_n$ to be
    \begin{align*}
        I_\psi(\PI)
        =\sum_{Q\in\PI}\tr[\dyad{\psi}Q]^2
        &=\sum_{Q\in\PI} \tr[\left(\frac{1}{2^n}(\one + \sum_{\one\neq P\in\aP_n} \alpha_P P)Q\right)^2] \\
        &= \sum_{Q\in\PI} \frac{1}{2^{2n}} \tr[\one] + \frac{1}{2^{2n}} \sum_{Q\in\PI} \sum_{\one\neq P,P'\in\aP_n} \alpha_P\alpha_{P'} \tr[PQP'Q] \\
        &= \frac{|\PI|}{2^n} + \frac{1}{2^{2n}} \sum_{Q\in\PI} \left(\sum_{\substack{\one\neq P\in\aP_n\\ P\in C(Q)}} \sum_{\one \neq P'\in\aP_n} \alpha_P\alpha_{P'} \tr[PQP'Q] + \sum_{\substack{\one\neq P\in\aP_n\\ P\notin C(Q)}}\sum_{\one \neq P'\in\aP_n} \alpha_P\alpha_{P'} \tr[PQP'Q] \right) \\
        &= \frac{|\PI|}{2^n} + \frac{1}{2^n} \sum_{Q\in\PI} \left(\sum_{\substack{\one\neq P\in\aP_n\\ P\in C(Q)}} \alpha^2_P - \sum_{\substack{\one\neq P\in\aP_n\\ P\notin C(Q)}} \alpha^2_P \right)\; ,
    \end{align*}
    where $C(Q)=\{P\in\aP_n\mid [P,Q]=0\}$ is the commutant of $Q$ in $\aP_n$. We thus obtain
    \begin{align*}
        I_\psi(\PI)
        &=\frac{|\PI|}{2^n} + \frac{1}{2^n} \sum_{Q\in\PI} \left(2\sum_{\substack{\one\neq P\in\aP_n\\ P\in C(Q)}} \alpha^2_P - \sum_{\one\neq P\in\aP_n} \alpha^2_P\right) \\
        &=\frac{1-2^{n-1}}{2^{n-1}} |\PI| + \frac{1}{2^{n-1}} \sum_{Q\in\PI} \sum_{\substack{\one\neq P\in\aP_n\\ P\in C(Q)}} \alpha^2_P\\
        &=\frac{1-2^{n-1}}{2^{n-1}} |\PI|+\frac{1}{2^{n-1}}\sum_{\one\neq P\in\aP_n} |C_\PI(P)|\alpha^2_P\; ,
    \end{align*}
    where $C_\PI(P)=C(P)\cap\PI$ and we used $\sum_{\one\neq P\in\aP_n}\alpha^2_P=2^n-1$ by the purity constraint $\tr[\rho^2]=1$ for $\rho=\dyad{\psi}$.
\end{proof}

\section{Proof of Lm.~\ref{lm: algebraic closure}}\label{app: algebraic closure}

\setcounter{lemma}{2}

\begin{lemma}
    $\cJ$ is closed under anti-commutators and $\fg$ is closed under commutators. Moreover, for $n\geq 2$,
    \begin{equation*}
    \begin{aligned}
        \{\cJ,\cJ\}&=\cJ &\quad\quad\quad
        \{\cJ,\fg\}&=\fg &\quad\quad\quad
        \{\fg,\fg\}&=\cJ \\
        [\cJ,\cJ]&=\fg &\quad\quad\quad
        [\cJ,\fg]&=\cJ &\quad\quad\quad
        [\fg,\fg]&=\fg
    \end{aligned}
    \end{equation*}
    whereas $[\cJ,\cJ]=[\cJ,\fg]=0$ for $n=1$.
\end{lemma}

\begin{proof}
    The case $\{\ta_X,\ta_Z\},\{\ta'_X,\ta'_Z\}\in\cJ$ with $[\{\ta_X,\ta_Z\},\{\ta'_X,\ta'_Z\}]=0$ was already treated in the main body. The remaining cases are analysed similarly, we list them here for completeness:
    \begin{itemize}
        \item Let $0\neq\{\ta_X,\ta_Z\},\{g^Z_X\ta'_X,g^X_Z\ta'_Z\}\in\cJ$, hence, $[\ta_X,\ta_Z]=0\neq[\ta'_X,\ta'_Z]$, and $[\ta_X\ta_Z,g^Z_X\ta'_Xg^X_Z\ta'_Z]=0$, which by Lm.~\ref{lm: nested commutativity} implies $[\ta_X,\ta'_Z]=0$ if and only if $[\ta'_X,\ta_Z]=0$. Consequently,
        \begin{align*}
            \{\{\ta_X,\ta_Z\},\{g^Z_X\ta'_X,g^X_Z\ta'_Z\}\}
            \sim
            \ta_X\ta_Zg^Z_X\ta'_Xg^X_Z\ta'_Z
            \sim
            g^Z_X\ta_X\ta'_Xg^X_Z\ta_Z\ta'_Z
            \sim\{g^Z_X\ta_X\ta'_X,g^X_Z\ta_Z\ta'_Z\}\in\cJ\; ,
        \end{align*}
        where the fourth step again follows with Lm.~\ref{lm: nested commutativity} given the above commutation relations.
        
        \item Let $0\neq\{\ta_X,\ta_Z\},[\ta'_X,g^X_Z\ta'_Z]\in\cJ$, hence, $[\ta_X,\ta_Z]=0\neq[\ta'_X,\ta'_Z]$, and $[\ta_X\ta_Z,\ta'_Xg^X_Z\ta'_Z]=0$, which by Lm.~\ref{lm: nested commutativity} implies $[\ta_X,\ta'_Z]=0$ if and only if $[\ta'_X,\ta_Z]=0$. Consequently,
        \begin{align*}
            \{\{\ta_X,\ta_Z\},[\ta'_X,g^X_Z\ta'_Z]\}
            \sim
            \ta_X\ta_Z\ta'_Xg^X_Z\ta'_Z
            \sim
            \ta_X\ta'_Xg^X_Z\ta_Z\ta'_Z
            \sim[\ta_X\ta'_X,g^X_Z\ta_Z\ta'_Z]\in\cJ\; ,
        \end{align*}
        where the fourth step again follows with Lm.~\ref{lm: nested commutativity} given the above commutation relations. The case $0\neq\{\ta_X,\ta_Z\},[g^Z_X\ta'_X,\ta'_Z]\in\cJ$ and $[\ta_X\ta_Z,g^Z_X\ta'_X\ta'_Z]=0$ is analogous, by swapping $X\longleftrightarrow Z$.
    
        \item Let $0\neq\{g^Z_X\ta_X,g^X_Z\ta_Z\},\{g^Z_X\ta'_X,g^X_Z\ta'_Z\}\in\cJ$, hence, $[\ta_X,\ta_Z]\neq0\neq[\ta'_X,\ta'_Z]$, and $[g^Z_X\ta_Xg^X_Z\ta_Z,g^Z_X\ta'_Xg^X_Z\ta'_Z]=0$ which by Lm.~\ref{lm: nested commutativity} implies $[\ta_X,\ta'_Z]=0$ if and only if $[\ta'_X,\ta_Z]=0$. Consequently,
        \begin{align*}
            \{\{g^Z_X\ta_X,g^X_Z\ta_Z\},\{g^Z_X\ta'_X,g^X_Z\ta'_Z\}\}
            \sim g^Z_X\ta_Xg^X_Z\ta_Zg^Z_X\ta'_Xg^X_Z\ta'_Z
            \sim \ta_X\ta'_X\ta_Z\ta'_Z
            \sim\{\ta_X\ta'_X,\ta_Z\ta'_Z\}\in\cJ\; ,
        \end{align*}
        where the fourth step again follows with Lm.~\ref{lm: nested commutativity} given the above commutation relations.
        
        \item Let $0\neq\{g^Z_X\ta_X,g^X_Z\ta_Z\},[\ta'_X,g^X_Z\ta'_Z]\in\cJ$, hence, $[\ta_X,\ta_Z]\neq0\neq[\ta'_X,\ta'_Z]$, and $[g^Z_X\ta_Xg^X_Z\ta_Z,\ta'_Xg^X_Z\ta'_Z]=0$, which by Lm.~\ref{lm: nested commutativity} implies $[\ta_X,\ta'_Z]=0$ if and only if $[\ta'_X,\ta_Z]\neq0$. Consequently,
        \begin{align*}
            \{\{g^Z_X\ta_X,g^X_Z\ta_Z\},[\ta'_X,g^X_Z\ta'_Z]\}
            \sim g^Z_X\ta_Xg^X_Z\ta_Z\ta'_Xg^X_Z\ta'_Z
            \sim g^Z_X\ta_X\ta'_X\ta_Z\ta'_Z
            \sim[g^Z_X\ta_X\ta'_X,\ta_Z\ta'_Z]\in\cJ\; ,
        \end{align*}
        where the fourth step again follows with Lm.~\ref{lm: nested commutativity} given the above commutation relations. The case $0\neq\{g^Z_X\ta_X,g^X_Z\ta_Z\},[g^Z_X\ta'_X,\ta'_Z]\in\cJ$ and $[g^Z_X\ta_Xg^X_Z\ta_Z,g^Z_X\ta'_X\ta'_Z]=0$ is analogous, by swapping $X\longleftrightarrow Z$.

        \item Let $0\neq[\ta_X,g^X_Z\ta_Z],[\ta'_X,g^X_Z\ta'_Z]\in\cJ$, hence, $[\ta_X,\ta_Z]\neq0\neq[\ta'_X,\ta'_Z]$, and $[\ta_Xg^X_Z\ta_Z,\ta'_Xg^X_Z\ta'_Z]=0$, which by Lm.~\ref{lm: nested commutativity} implies $[\ta_X,\ta'_Z]=0$ if and only if $[\ta'_X,\ta_Z]=0$. Consequently,
        \begin{align*}
            \{[\ta_X,g^X_Z\ta_Z],[\ta'_X,g^X_Z\ta'_Z]\}
            \sim \ta_Xg^X_Z\ta_Z\ta'_Xg^X_Z\ta'_Z
            \sim \ta_X\ta'_X\ta_Z\ta'_Z
            \sim\{\ta_X\ta'_X,\ta_Z\ta'_Z\}\in\cJ\; ,
        \end{align*}
        where the fourth step again follows with Lm.~\ref{lm: nested commutativity} given the above commutation relations. The case $0\neq[g^Z_X\ta_X,\ta_Z],[g^Z_X\ta'_X,\ta'_Z]\in\cJ$ and $[g^Z_X\ta_X\ta_Z,g^Z_X\ta'_X\ta'_Z]=0$ is analogous, by swapping $X\longleftrightarrow Z$.

        \item Lastly, let $0\neq[\ta_X,g^X_Z\ta_Z],[g^Z_X\ta'_X,\ta'_Z]\in\cJ$, hence, $[\ta_\cw{X},\ta_Z]\neq0\neq[\ta'_X,\ta'_Z]$, and $[\ta_Xg^X_Z\ta_Z,g^Z_X\ta'_X\ta'_Z]=0$, which by Lm.~\ref{lm: nested commutativity} implies $[\ta_X,\ta'_Z]=0$ if and only if $[\ta'_X,\ta_Z]\neq0$. Consequently,
        \begin{align*}
            \{[\ta_X,g^X_Z\ta_Z],[g^Z_X\ta'_X,\ta'_Z]\}
            \sim \ta_Xg^X_Z\ta_Zg^Z_X\ta'_X\ta'_Z
            \sim g^Z_X\ta_X\ta'_Xg^X_Z\ta_Z\ta'_Z
            \sim\{g^Z_X\ta_X\ta'_X,g^X_Z\ta_Z\ta'_Z\}\in\cJ\; ,
        \end{align*}
        where the fourth step again follows with Lm.~\ref{lm: nested commutativity} given the above commutation relations.
    \end{itemize}
    Taken together, this proves $\ta,\ta'\in\cJ$, $[\ta,\ta']=0$ $\Rightarrow$ $\{\ta,\ta'\}\sim\ta\ta'\in\cJ$, that is, $\{\cJ,\cJ\}\subset\cJ$. Since the above arguments only use the parity of commutation relations and since $L=\aP_n\backslash J$ has a decomposition as in Eq.~(\ref{eq: family of purity invariants}) with commutators and anti-commutators exchanged everywhere, the same reasoning also shows $[\cJ,\cJ]\subset\fg$ as well as $\{\fg,\fg\}\subset\cJ$, $[\fg,\fg]\subset\fg$ and $\{\cJ,\fg\}\subset\fg$, $[\cJ,\fg]\subset\cJ$. Consider, for instance, the case that $0\neq[\ta_X,\ta_Z],[\ta'_X,\ta'_Z]\in L$ with $[[\ta_X,\ta_Z],[\ta'_X,\ta'_Z]]\neq0$, hence, $\{\ta_X,\ta_Z\}=0=\{\ta'_X,\ta'_Z\}$, and $\{[\ta_X,\ta_Z],[\ta'_X,\ta'_Z]\}=0$. By Lm.~\ref{lm: nested commutativity}, the latter is equivalent to $[\ta_X,\ta'_Z]=0$ if and only if $[\ta'_X,\ta_Z]\neq0$, consequently
    \begin{align*}
        [[\ta_X,\ta_Z],[\ta'_X,\ta'_Z]]
        \sim[\ta_X\ta_Z,\ta'_X\ta'_Z]
        \sim
        \ta_X\ta_Z\ta'_X\ta'_Z
        \sim
        \ta_X\ta'_X\ta_Z\ta'_Z
        \sim[\ta_X\ta'_X,\ta_Z\ta'_Z]\in L\; ,
    \end{align*}
    where we used $[\ta_X\ta'_X,\ta_Z\ta'_Z]\neq0$ in the fourth step, which follows from Lm.~\ref{lm: nested commutativity} given the above commutation relations.

    Finally, it follows from the explicit form of $\cJ$ in Eq.~(\ref{eq: family of purity invariants}), and the product relations above, that these inclusions are in fact equalities for $n\geq 2$. This is immediate for the subset $J_{\{\tA_X,\tA_Z\}}$ in Eq.~(\ref{eq: family of purity invariants}). For the others it follows by reading the case analysis of anti-commutators above in reverse. For instance, take the element $0\neq [\ta_X,g^X_Z\ta_Z]\sim\ta_Xg^X_Z\ta_Z\in J$ and let $\ta_X=\ta'_X\ta''_X$ and $\ta_Z=\ta'_Z\ta''_Z$ (with one factor possibly equal to the identity). Since $[\ta_X,\ta_Z]\neq0$, we may wlog assume that $[\ta'_X,\ta'_Z]\neq 0$, in which case either $\{\{\ta''_X,\ta''_Z\},[\ta'_X,g^X_Z\ta'_Z]\}\sim\ta_Xg^X_Z\ta_Z$ if  $[\ta''_X,\ta''_Z]=0$ or $\{\{g^Z_X\ta_X,g^X_Z\ta_Z\},[g^Z_X\ta'_X,\ta'_Z]\}\sim\ta_Xg^X_Z\ta_Z$ if $[\ta''_X,\ta''_Z]\neq0$. (In fact, it is sufficient to pick $\ta''_X=\one=\ta''_Z$.)

    What is more, for $n\geq 3$, one also finds $\{\cJ^\circ,\cJ^\circ\}=\cJ^\circ$. For $n=2$, $\tA_X=\langle\ta_X\rangle,\tA_Z=\langle\ta_Z\rangle$ have order $1$, and
    \begin{equation}\label{eq: n=2}
    \begin{aligned}
        \cJ^\circ
        &=\{\ta_X,\ta_Z,\{g^Z_X\ta_X,g^X_Z\ta_Z\},[g^Z_X\ta_X,\ta_Z],[\ta_X,g^X_Z\ta_Z]\}\; , \\
        \fg
        &=\{g^Z_X,g^Z_X\ta_X,g^X_Z,g^X_Z\ta_Z,
        [\ta_X,\ta_Z],[g^Z_X,g^X_Z],[g^Z_X\ta_X,g^X_Z],[g^Z_X,g^X_Z\ta_Z],\{g^Z_X,\ta_Z\},\{\ta_X,g^X_Z\}\}\; .
    \end{aligned}
    \end{equation}
    Since $[\ta_X,\ta_Z]\neq 0$ in this case, $\cJ^\circ$ defines a subset of mutually anti-commuting Pauli operators, and thus $\{\cJ^\circ,\cJ^\circ\}=\{\one\}$ in this case. The relations in the single-qubit case ($n=1$) where $\cJ=\{\one\}$ and $\fg=\{X,Y,Z\}$ are obvious.
\end{proof}

Notably, the closure relations (under anti-commutators) in Lm.~\ref{lm: algebraic closure}, only hold trivially for $n=2$. More precisely, $\{\cJ^\circ,\cJ^\circ\}=\{\one\}$ (with $\cJ^\circ=\cJ\backslash\{\one\}$), whereas $\{\cJ^\circ,\cJ^\circ\}=\cJ^\circ$ holds for $n>2$. Together with the fact that $\cJ^\circ$ is not a set of mutually anti-commutating Hermitian Pauli operators for $n>2$ (see Sec.~\ref{sec: naive generalisation}), conceals the algebraic structure of the sets Def.~\ref{def: family of purity invariants} in this case $n=2$; in turn, this is the reason why the generalisation from $n=2$ to $n>2$ is nontrivial.

\section{Proofs of Lm.~\ref{lm: no max ab subgroup}}\label{app: no max ab subgroup}

\setcounter{lemma}{3}

\begin{lemma}
    $\cJ$ contains no Abelian subgroup of dimension larger than $n-1$.
\end{lemma}

\begin{proof}
    We define the indicator function $\varphi:\aP_n\ra\zz_2$ for membership in the respective parts of the decomposition $\aP_n=\cJ\oplus\fg$ by  $\varphi(P)=0$ if and only if $P\in\cJ$. Eq.~(\ref{eq: algebraic closure relations}) is then equivalent to the product formula $\varphi(PQ)=\varphi(P)+\varphi(Q)+\omega(P,Q)\mod 2$, where $\omega(Q,P)=0$ if and only if $[P,Q]=0$.\footnote{In other words, $\varphi$ in $\phi=(-1)^{\varphi}$ is a quadratic refinement of $\omega$, that is, $\varphi(QP)=\varphi(P)+\varphi(Q)+\omega(P,Q)$ (cf. Ref.~\cite{Frembs2026}).\label{fn: quadratic refinement}}
    
    Let $A,A'<\aP_n$ be two maximal Abelian subgroups such that $A\cap A'=\{\one\}$. Then there exist generating sets $\cG=\{Q_i\}_{i=1}^n$, $\cG'=\{Q'_k\}_{k=1}^n$ with $A=\langle\cG\rangle$, $A'=\langle\cG'\rangle$ and such that $\omega(Q_i,Q'_j)=\delta_{ij}$ (and $\omega(Q_i,Q_j)=0=\omega(Q'_i,Q'_j)$ for all $i,j\in[n]$) (see also Lm.~\ref{lm: conjugate pairs}).\footnote{In other words, $\cG$, $\cG'$ define a symplectic basis in the symplectic vector space $\zz^{2n}_2\cong\aP_n$.} With respect to this parametrisation, we define the following expression\footnote{$\Delta(\varphi)$ is called the \emph{Arf invariant} of a quadratic form $\varphi$ in characteristic 2 and, together with its dimension, determines it uniquely (up to equivalence) \cite{Arf1941,Dye1978,Knebusch2010}. For more details on the relation between quadratic forms, their Arf invariant and the relation with the characterisation of purity invariants, see Ref.~\cite{Frembs2026}.}
    \begin{equation}\label{eq: Gauss sum}
    \begin{aligned}
        \Delta(\phi)
        =\sum_{P\in\aP_n}\phi(P)
        =\sum_{P\in\aP_n}(-1)^{\varphi(P)}
        &=\sum_{P\in\aP_n}\prod_{i=1}^n(-1)^{\varphi(P_i)}\\
        &=\prod_{i=1}^n\sum_{P\in\langle Q_i,Q'_i\rangle}(-1)^{\varphi(P)}\\
        &=\prod_{i=1}^n\left(1+(-1)^{\varphi(Q_i)}+(-1)^{\varphi(Q'_i)}+(-1)^{\varphi(Q_iQ'_i)}\right)\\
        &=\prod_{i=1}^n\left(1+(-1)^{\varphi(Q_i)}+(-1)^{\varphi(Q'_i)}+(-1)^{\varphi(Q_i)+\varphi(Q'_i)+1}\right)\\
        &=\prod_{i=1}^n 2(-1)^{\varphi(Q_i)\varphi(Q'_i)}
        =2^n(-1)^{\sum_{i=1}^n\varphi(Q_i)\varphi(Q'_i)}\; ,
    \end{aligned}
    \end{equation}
    where we used that, according to the above product formula, $\varphi$ factorises with respect to the decomposition of Pauli operators into local ones $P=\prod_{i=1}^nP_i$ (since the symplectic inner product vanishes between operators supported on different factors).\footnote{More generally, $\varphi$ factorises with respect to any decomposition of $\aP_n$ into orthogonal symplectic subspaces.} In particular, up to Clifford conjugation (see Lm.~\ref{lm: conjugate generators}), we may express Pauli operators with respect to the canonical choice of Abelian subgroups in Def.~\ref{def: family of purity invariants}, for which $A_X=\langle X_1,\cdots,X_n\rangle$, $\tA_X=\langle X_1,\cdots,X_{n-1}\rangle$ and $A_Z=\langle Z_1,\cdots,Z_n\rangle$, $\tA_Z=\langle Z_1,\cdots,Z_{n-1}\rangle$. For this choice, Eq.~(\ref{eq: Gauss sum}) evaluates to $\Delta(\phi)=\sum_{P\in\aP_n}(-1)^{\varphi(P)}=2^n(-1)^{\sum_{i=1}^n\varphi(X_i)\varphi(Z_i)}=-2^n$ as $\varphi(X_i)=\varphi(Z_i)=0$ for $i<n$ and $=1$ for $i=n$. Moreover, $\Delta(\phi)$ is in fact (as our notation suggests) independent of the choice of disjoint maximal Abelian subgroups $A,A'$ \cite{Dye1978}.\footnote{Alternatively, the invariance of $\Delta(\phi)$ also follows from the Lm.~\ref{lm: Clifford covariance}, and the fact that cardinalities remain invariant under Clifford conjugation.} To see this, note that by Lm.~\ref{lm: conjugate pairs}, it is sufficient to prove that $\sum_{i=1}^n\varphi(CX_iC^\dagger)\varphi(CZ_iC^\dagger)=\sum_{i=1}^n\varphi(X_i)\varphi(Z_i)$ for any Clifford unitary $C\in\Cl_n$, specifically for any of the generators in App.~\ref{app: Abelian Pauli group}. It is easy to see that SWAP and Hadamard gates leave the sum unchanged, as they only permute site indices or exchange $X_i\leftrightarrow Z_i$. Moreover, for CNOT gates $\Lambda\!\mathrm{X}_{ij}$ with action given by Eq.~(\ref{eq: CNOT}), invariance follows since
    \begin{align*}
        &\varphi(\Lambda\!\mathrm{X}_{ij}X_i\Lambda\!\mathrm{X}_{ij})\varphi(\Lambda\!\mathrm{X}_{ij}Z_i\Lambda\!\mathrm{X}_{ij})+\varphi(\Lambda\!\mathrm{X}_{ij}X_j\Lambda\!\mathrm{X}_{ij})\varphi(\Lambda\!\mathrm{X}_{ij}Z_j\Lambda\!\mathrm{X}_{ij})\\
        =\ &\varphi(X_iX_j)\varphi(Z_i)+\varphi(X_j)\varphi(Z_iZ_j)\\
        =\ &(\varphi(X_i)+\varphi(X_j))\varphi(Z_i)+\varphi(X_j)(\varphi(Z_i)+\varphi(Z_j))\\
        =\ &\varphi(X_i)\varphi(Z_i)+\varphi(X_j)\varphi(Z_j)\mod 2\; ,
    \end{align*}
    where we used the product formula for $\varphi$ (from Eq.~(\ref{eq: algebraic closure relations})) in the second step. Similarly, for phase gates $S_i$, one finds
    \begin{align*}
        \varphi(S_iX_iS^\dagger_i)\varphi(S_iZ_iS^\dagger_i)
        =\varphi(X_iZ_i)\varphi(Z_i)
        =(\varphi(X_i)+\varphi(Z_i)+1)\varphi(Z_i)
        =\varphi(X_i)\varphi(Z_i)\mod 2\; .
    \end{align*}
    Now, assume that $B<\cJ$ is a maximal Abelian subgroup with generating set $\cH=\{R_i\}_{i=1}^n$, and let $B'<\aP_n$ be another maximal Abelian subgroup with $B\cap B'=\{\one\}$ and corresponding generating set $\cH'=\{R'_i\}_{i=1}^n$ such that $\omega(R_i,R'_j)=\delta_{ij}$. Since $B<\cJ$, $\varphi|_B=0$ and thus $\Delta(\phi)=\sum_{P\in\aP_n}(-1)^{\varphi(P)}=2^n(-1)^{\sum_{i=1}^n\varphi(R_i)\varphi(R'_i)}=2^n$. Evidently, this does not match the value found above, and we thus conclude that $\cJ$ contains no maximal Abelian subgroup.
\end{proof}

\section{Proofs of Lm.~\ref{lm: constant commutants}}\label{app: constant commutants}

\setcounter{lemma}{4}

\begin{lemma}
    The cardinalities of the commutants $C_\cJ(P\in\cJ)$ and $C_\fg(P\in\cJ)$ are independent of $\one\neq P\in\cJ$, while the cardinality of the commutants $C_\cJ(P\in\fg)$ and $C_\fg(P\in\fg)$ are independent of $P\in\fg$.
\end{lemma}

\begin{proof}
    For $n=1$, $\cJ=\{\one\}$ and the statement is immediately verified. For the general case, we identify $\cJ$ and $\fg$ with the level sets of the indicator function $\varphi:\aP_n\ra\zz_2$ (as in the proof of Lm.~\ref{lm: no max ab subgroup}) and define
    \begin{align*}
        \phi(P)
        =(-1)^{\varphi(P)}
        =\begin{cases}
            +1&\mathrm{if}\ P\in\cJ\\
            -1&\mathrm{if}\ P\in\fg
        \end{cases}\; .
    \end{align*}
    The closure relations in Eq.~(\ref{eq: algebraic closure relations}) are then equivalent to the condition $\phi(PQ)=(-1)^{\omega(P,Q)}\phi(P)\phi(Q)$, where $\omega(P,Q)=0$ if and only if $[P,Q]=0$.
    
    Now, for any $P\in\aP_n$, define $\Delta(P):=\sum_{Q\in C(P)}\phi(Q)=|C_\cJ(P)|-|C_\fg(P)|$. Moreover, we compute for any fixed $P\in\aP_n$
    \begin{align*}
        \Delta(\phi)
        :=\sum_{Q\in\aP_n}\phi(Q)
        &=\sum_{Q\in\aP_n}\phi(PQ)\\
        &=\phi(P)\sum_{Q\in\aP_n}(-1)^{\omega(P,Q)}\phi(Q)\\
        &=\phi(P)\left(\sum_{Q\in C(P)}\phi(Q)-\sum_{Q\notin C(P)}\phi(Q)\right)
        =\phi(P)(\Delta(P)-(\Delta(\phi)-\Delta(P)))\; ,
    \end{align*}
    where we used the above identity for $\phi$ in the third step. Clearly, this yields $\Delta(P)=\Delta(\phi)$ for $P\in\cJ$ and $\Delta(P)=0$ for $P\in\fg$. Since $\Delta(\phi)$ is a constant and every non-identity Pauli operator (anti-)commutes with half of the elements in $\aP_n$, $|C(P)|=4^n/2$ the desired cardinalities are thus easily seen to be independent of $P\in\fg$, and of $\one\neq P\in\cJ$, since
    \begin{align*}
    \begin{split}
        |C_\cJ(P)|
        &=\frac{1}{2}(|C(P)|+\Delta(P))
        =4^{n-1}+\frac{1}{2}\Delta(P) \\
        |C_\fg(P)|
        &=\frac{1}{2}(|C(P)|-\Delta(P))
        =4^{n-1}-\frac{1}{2}\Delta(P)
    \end{split}\; .
    \end{align*}
    Explicitly, since $|\cJ|=\frac{1}{2}4^n-2^{n-1}$ by Lm.~(\ref{lm: cardinalities}) we find $\Delta(\phi)=\sum_{Q\in\aP_n}\phi(Q)=|\cJ|-|\fg|=-2^n$, and thus $|C_\cJ(P)|=2^{n-1}(2^{n-1}-1)$ and $|C_\fg(P)|=2^{n-1}(2^{n-1}+1)$ for all $\one\neq P\in\cJ$, as well as $|C_\cJ(P)|=|C_\fg(P)|=4^{n-1}$ for all $P\in\fg$.
\end{proof}

\section{Proof of Lm.~\ref{lm: Clifford covariance}}\label{app: Clifford invariance}

\setcounter{lemma}{5}

\begin{lemma}
    The various sets in Def.~\ref{def: family of purity invariants} define an orbit under conjugation by the Clifford group.
\end{lemma}

\begin{proof}
    Clearly, if $\PI\subsetneq\aP_n$ defines a purity invariant, $I_\psi(\PI)=\mathrm{const}$, then so does $U\PI U^{-1}$ for any unitary $U\in\cU(\C^{2^n})$; and, by definition of the Clifford group, $U\PI U^{-1}\subset\aP_n$ for $U\in\Cl_n$. However, a priori there may exist non-Clifford unitaries $U\in\cU(\C^{2^n})$ for which $U\PI U^{-1}\subset\cP_n$; moreover, purity invariants in Def.~\ref{def: family of purity invariants} may decompose into several unitary (in particular, Clifford) orbits. Yet, neither of this is the case, since, by Lm.~\ref{lm: algebraic closure}, $U\PI U^{-1}\subset\cP_n$ implies $\langle U\PI U^{-1}\rangle \subset\langle\cP_n\rangle \Leftrightarrow  U\langle\PI\rangle  U^{-1}\subset\cP_n \Leftrightarrow U\cP_n U^{-1}\subset\cP_n$, hence, $U$ must be Clifford, and by Lm.~\ref{lm: conjugate generators}, the defining data of the sets $J$ in Def.~\ref{def: family of purity invariants} (that is, a pair of maximal Abelian subgroups $A_Z,A_X$ with $A_X\cap A_Z=\{\one\}$, subgroups $\tA_Z<A_Z$, $\tA_X<A_X$ of dimension $n-1$ and elements $A_Z\ni g^X_Z\notin\tA_Z$, $A_X\ni g^Z_X\notin\tA_X$ such that $[g^X_Z,\tA_X]=0=[g^Z_X,\tA_Z]$) can be conjugated onto any other by a Clifford unitary.
\end{proof}

\section{Proof of Thm.~\ref{thm: counting}}\label{app: counting}

\setcounter{theorem}{1}

\begin{theorem}
    There exist $2^{n-1}(2^n-1)$ distinct (Clifford-conjugate) subsets $\cJ$ in $\aP_n$.
\end{theorem}

\begin{proof}
    Let $A_X,A_Z<\aP_n$ be two maximal Abelian subgroups with $A_X\cap A_Z=\{\one\}$.
    
    We first show that there are $2^n-1$ ways to choose subgroups $\tA_X\leq A_X$ of dimension $n-1$. To this end, we choose a minimal generating set (of size $n-1$) for such $\tA_X$. For the first element, there are $2^n-1$ possible choices, which reduces the possibilities for the next element by $2^1-1$. Choosing the first two elements reduces the possibilities for the third by $2^2-1$ and so on until we have $2^n-1-(2^{n-2}-1)$ choices for the $(n-1)$-th element. Clearly, this procedure yields a minimal generating set for subgroups $\tA_X<A_X$. However, it is not unique, and we thus need to divide by the number of different generating sets for a fixed $\tA_X<A_X$. Following the same reasoning, for the first element, we may choose any of the $2^{n-1}-1$ elements, for the second, there are $2^{n-1}-1-(2^1-1)$ possibilities, and so on until for the $(n-1)$-th there are $2^{n-1}-1-(2^{n-2}-1)$ choices. The number of different subgroups $\tA_X<A_X$ of dimension $n-1$ thus reads
    \begin{align*}
        \frac{\prod_{i=0}^{n-2}(2^n-1-(2^i-1))}{\prod_{i=0}^{n-2}(2^{n-1}-1-(2^i-1))}
        =\frac{(2^n-1)\prod_{i=1}^{n-2}2(2^{n-1}-2^{i-1})}{(2^{n-1}-2^{n-2})\prod_{i=0}^{n-3}(2^{n-1}-2^i)} 
        =\frac{(2^n-1)2^{n-2}\prod_{i=0}^{n-3}(2^{n-1}-2^{i})}{2^{n-2}\prod_{i=0}^{n-3}(2^{n-1}-2^i)}
        =2^n-1.
    \end{align*}

    Next, we show that there are $2^{n-1}$ ways to choose a subgroup $\tA_Z<A_Z$ of dimension $n-1$, that is compatible with the conditions in Def.~\ref{def: family of purity invariants}. These include that there exists a unique element $g^X_Z\in A_Z$ such that $[g^X_Z,\tA_X]=0$ (uniqueness follows since $\langle\tA_X,g^X_Z\rangle$ is a maximal Abelian subgroup), which has to be excluded from $\tA_Z$. Using the same reasoning as above, the number of subgroups of dimension $n-1$ that include $g^X_Z$ is easily computed to be
    \begin{equation*}
        \frac{\prod_{i=1}^{n-2}(2^n-1-(2^i-1))}{\prod_{i=1}^{n-2}(2^{n-1}-1-(2^i-1))}
        =\frac{2^{n-1}-1}{2^n-1}\left(\frac{\prod_{i=0}^{n-2}(2^n-1-(2^i-1))}{\prod_{i=0}^{n-2}(2^{n-1}-1-(2^i-1))}\right)
        =2^{n-1}-1\; .
    \end{equation*}
    This then leaves $2^n-1-(2^{n-1}-1)=2^{n-1}$ ways to choose subgroups $\tA_Z<A_Z$ of dimension $n-1$ that do not contain $g^X_Z$. (Clearly, different choices of $\tA_X,\tA_Z$ generate different sets $J$ since $\tA_X=J\cap A_X$ and $\tA_Z=J\cap A_Z$ and we are thus left with $2^{n-1}(2^n-1)$ distinct subsets $\cJ$.)

    Finally, by Lm.~\ref{lm: conjugate pairs}, the choice of subgroups $A_X,A_Z$ is without loss of generality, and the result follows. To see this, we show that every set of the form in Def.~\ref{def: family of purity invariants} can be parametrised with respect to the reference subgroups $A_X,A_Z$. Let $J'=J'(\tA'_X,g'^Z_X;A'_X,A'_Z)$ be any set as in Def~\ref{def: family of purity invariants}, defined with respect to any quadruple $(\tA'_X,g'^Z_X;A'_X,A'_Z)$. By Lm.~\ref{lm: algebraic closure}, and since $J'$ is a purity invariant with $I_\psi(J')=2^{n-1}$, $J'$ intersects every maximal Abelian subgroup in an Abelian subgroup of rank $n-1$. With respect to the above reference pair of maximal Abelian subgroups $A_X, A_Z$ we thus obtain subgroups $A_X\cap J'=:\tA_X<A_X$ and $A_Z\cap J'=:\tA_Z<A_Z$ which further uniquely define the elements $g^Z_X\in A_X$ and $g^X_Z\in A_Z$, that due to Lm.~\ref{lm: no max ab subgroup} satisfy $g^Z_X, g^X_Z\in \fg'=\aP_n\backslash\cJ'$. Let $J=J(\tA_X,g^Z_X;A_X,A_Z)$ be the corresponding set in Def.~\ref{def: family of purity invariants}. By (repeated use of) Lm.~\ref{lm: algebraic closure}, $J'$ must contain the respective (anti-)commutators in the definition of $J$, and since this uniquely defines this set, we conclude $J'=J$.
\end{proof}

\section{purity invariants for three qubits}\label{app: n=3}

Invariants in Def.~\ref{def: family of purity invariants} are trivial for $n=1$, and for $n=2$ correspond with the six `pentagons' (maximal sets of mutually anti-commuting Pauli operators) in Ref.~\cite{HoehnWever2017}. Below, we list the 28 subsets of the form in Def.~\ref{def: family of purity invariants} for $n=3$, see the counting result in Thm.~\ref{thm: counting}. (As before, we omit tensor products and write the factors in Pauli operators by simple concatenation. Below, we also do not list the identity element $\one\in\aP_n$, which must be added to identify the sets below with the family in Def.~\ref{def: family of purity invariants}.) With Lm.~\ref{lm: Clifford covariance}, these can be obtained by Clifford conjugating any given set $J$, e.g. the canonical set $J$ corresponding in Lm.~\ref{lm: conjugate generators}. For each of these subsets, we highlight $\tA_X<A_X=\langle X_i\rangle_{i=1}^3$ (in green within the element list), $\tA_Z<A_Z=\langle Z_i\rangle_{i=1}^3$ (in red within the element list), $g_X^Z$ (in green in the rightmost column), $g_Z^X$ (in red in the rightmost column) which can be used in Def.~\ref{def: family of purity invariants} to generate the relevant subset.

In Sec.~\ref{sec: family of pure state invariants}, we argued that the complementarity structure of the pentagon identities in the two-qubit case can equivalently be understood in terms of partitions of Pauli operators into $2^n+1$ maximal Abelian subgroups, equivalently in terms of the mutually unbiased bases of their common stabiliser eigenstates \cite{LawrenceBruknerZeilinger2002}. In this view, mutual anti-commutation between the Pauli operators in $\pent$ is encoded in terms of these elements belonging to different maximal Abelian subgroups of such a partition, and Def.~\ref{def: family of purity invariants} was (in part) motivated by generalising this latter view on complementarity (see also Ref.~\cite{FrembsHoehn2026b}). Indeed, the corresponding complementarity structure can be seen explicitly in the $J$-sets for three qubits below: each set contains $27$ non-identity Pauli elements, which with respect to any partition into $9$ maximal Abelian subgroups decompose into $9$ sets of $3$ (mutually commuting) elements in every subgroup.

\begin{figure}
    \centering
    \normalsize{
    \setlength{\tabcolsep}{4pt}
    \resizebox{\textwidth}{!}{%
    \begin{tabular}{|cccccccccccccc|c|} 
    \hline
    \multicolumn{14}{|c|}{Pauli strings of the 28 Clifford-conjugate purity invariants in Def.~\ref{def: family of purity invariants} for $n=3$} & \cx{$g_X^Z$},\cz{$g_Z^X$} \\
    \hline
    \multicolumn{1}{|c|}{1} & \cx{$\cx{\one\one X}$} & $\one\one Y$ & \cx{$\cx{\one X\one}$} & \cx{$\cx{\one XX}$} & $\one XY$ & $\one Y\one$ & $\one YX$ & $\one YY$ & \cz{$\cz{\one ZZ}$} & $X\one Z$ & $XXZ$ & $XYZ$ & $XZ\one$ & \cx{$XXX$} \\
    \cline{1-1}
    $XZX$ & $XZY$ & $Y\one Z$ & $YXZ$ & $YYZ$ & $YZ\one$ & $YZX$ & $YZY$ & \cz{$\cz{Z\one Z}$} & $ZXZ$ & $ZYZ$ & \cz{$\cz{ZZ\one}$} & $ZZX$ & $ZZY$ & \cz{$Z\one\one$}\\
   \hline

\multicolumn{1}{|c|}{2} & $\cx{\one\one X}$ & $\one\one Y$ & $\cx{\one X\one}$ & $\cx{\one XX}$ & $\one XY$ & $\one YZ$ & $\cz{\one Z\one}$ & $\one ZX$ & $\one ZY$ & $X\one Z$ & $XXZ$ & $XY\one$ & $XYX$ & $\cx{X\one X}$ \\
    \cline{1-1}
    $XYY$ & $XZZ$ & $Y\one Z$ & $YXZ$ & $YY\one$ & $YYX$ & $YYY$ & $YZZ$ & $\cz{Z\one Z}$ & $ZXZ$ & $ZY\one$ & $ZYX$ & $ZYY$ & $\cz{ZZZ}$ & $\cz{Z\one\one}$ \\
    \hline

\multicolumn{1}{|c|}{3} & $\cx{\one\one X}$ & $\one\one Y$ & $\one XZ$ & $\one Y\one$ & $\one YX$ & $\one YY$ & $\cz{\one Z\one}$ & $\one ZX$ & $\one ZY$ & $X\one Z$ & $\cx{XX\one}$ & $\cx{XXX}$ & $XXY$ & $\cx{X\one X}$ \\
    \cline{1-1}
    $XYZ$ & $XZZ$ & $Y\one Z$ & $YX\one$ & $YXX$ & $YXY$ & $YYZ$ & $YZZ$ & $\cz{Z\one Z}$ & $ZX\one$ & $ZXX$ & $ZXY$ & $ZYZ$ & $\cz{ZZZ}$ & $\cz{ZZ\one}$ \\
    \hline

\multicolumn{1}{|c|}{4} & $\cx{\one\one X}$ & $\one\one Y$ & $\one XZ$ & $\one YZ$ & $\cz{\one ZZ}$ & $\cx{X\one\one}$ & $\cx{X\one X}$ & $X\one Y$ & $XXZ$ & $XYZ$ & $XZZ$ & $Y\one\one$ & $Y\one X$ & $\cx{XXX}$ \\
    \cline{1-1}
    $Y\one Y$ & $YXZ$ & $YYZ$ & $YZZ$ & $\cz{Z\one Z}$ & $ZX\one$ & $ZXX$ & $ZXY$ & $ZY\one$ & $ZYX$ & $ZYY$ & $\cz{ZZ\one}$ & $ZZX$ & $ZZY$ & $\cz{\one Z\one}$ \\
    \hline

\multicolumn{1}{|c|}{5} & $\cx{\one\one X}$ & $\one\one Y$ & $\one XZ$ & $\one YZ$ & $\cz{\one ZZ}$ & $\cx{X\one\one}$ & $\cx{X\one X}$ & $X\one Y$ & $XXZ$ & $XYZ$ & $XZZ$ & $Y\one Z$ & $YX\one$ & $\cx{\one XX}$ \\
    \cline{1-1}
    $YXX$ & $YXY$ & $YY\one$ & $YYX$ & $YYY$ & $YZ\one$ & $YZX$ & $YZY$ & $\cz{Z\one\one}$ & $Z\one X$ & $Z\one Y$ & $ZXZ$ & $ZYZ$ & $\cz{ZZZ}$ & $\cz{\one Z\one}$ \\
    \hline

\multicolumn{1}{|c|}{6} & $\cx{\one\one X}$ & $\one\one Y$ & $\one XZ$ & $\one YZ$ & $\cz{\one ZZ}$ & $X\one Z$ & $\cx{XX\one}$ & $\cx{XXX}$ & $XXY$ & $XY\one$ & $XYX$ & $XYY$ & $XZ\one$ & $\cx{\one XX}$ \\
    \cline{1-1}
    $XZX$ & $XZY$ & $Y\one\one$ & $Y\one X$ & $Y\one Y$ & $YXZ$ & $YYZ$ & $YZZ$ & $\cz{Z\one\one}$ & $Z\one X$ & $Z\one Y$ & $ZXZ$ & $ZYZ$ & $\cz{ZZZ}$ & $\cz{ZZ\one}$ \\
    \hline

\multicolumn{1}{|c|}{7} & $\cx{\one\one X}$ & $\cz{\one\one Z}$ & $\cx{\one X\one}$ & $\cx{\one XX}$ & $\one XZ$ & $\one Y\one$ & $\one YX$ & $\one YZ$ & $\one ZY$ & $X\one Y$ & $XXY$ & $XYY$ & $XZ\one$ & $\cx{XX\one}$ \\
    \cline{1-1}
    $XZX$ & $XZZ$ & $Y\one Y$ & $YXY$ & $YYY$ & $YZ\one$ & $YZX$ & $YZZ$ & $Z\one Y$ & $ZXY$ & $ZYY$ & $\cz{ZZ\one}$ & $ZZX$ & $\cz{ZZZ}$ & $\cz{Z\one\one}$ \\
    \hline

\multicolumn{1}{|c|}{8} & $\cx{\one\one X}$ & $\cz{\one\one Z}$ & $\cx{\one X\one}$ & $\cx{\one XX}$ & $\one XZ$ & $\one YY$ & $\cz{\one Z\one}$ & $\one ZX$ & $\cz{\one ZZ}$ & $X\one Y$ & $XXY$ & $XY\one$ & $XYX$ & $\cx{X\one\one}$ \\
    \cline{1-1}
    $XYZ$ & $XZY$ & $Y\one Y$ & $YXY$ & $YY\one$ & $YYX$ & $YYZ$ & $YZY$ & $Z\one Y$ & $ZXY$ & $ZY\one$ & $ZYX$ & $ZYZ$ & $ZZY$ & $\cz{Z\one\one}$ \\
    \hline

\multicolumn{1}{|c|}{9} & $\cx{\one\one X}$ & $\cz{\one\one Z}$ & $\one XY$ & $\one Y\one$ & $\one YX$ & $\one YZ$ & $\cz{\one Z\one}$ & $\one ZX$ & $\cz{\one ZZ}$ & $X\one Y$ & $\cx{XX\one}$ & $\cx{XXX}$ & $XXZ$ & $\cx{X\one\one}$ \\
    \cline{1-1}
    $XYY$ & $XZY$ & $Y\one Y$ & $YX\one$ & $YXX$ & $YXZ$ & $YYY$ & $YZY$ & $Z\one Y$ & $ZX\one$ & $ZXX$ & $ZXZ$ & $ZYY$ & $ZZY$ & $\cz{ZZ\one}$ \\
    \hline

\multicolumn{1}{|c|}{10} & $\cx{\one\one X}$ & $\cz{\one\one Z}$ & $\one XY$ & $\one YY$ & $\one ZY$ & $\cx{X\one\one}$ & $\cx{X\one X}$ & $X\one Z$ & $XXY$ & $XYY$ & $XZY$ & $Y\one\one$ & $Y\one X$ & $\cx{XX\one}$ \\
    \cline{1-1}
    $Y\one Z$ & $YXY$ & $YYY$ & $YZY$ & $Z\one Y$ & $ZX\one$ & $ZXX$ & $ZXZ$ & $ZY\one$ & $ZYX$ & $ZYZ$ & $\cz{ZZ\one}$ & $ZZX$ & $\cz{ZZZ}$ & $\cz{\one Z\one}$ \\
    \hline

\multicolumn{1}{|c|}{11} & $\cx{\one\one X}$ & $\cz{\one\one Z}$ & $\one XY$ & $\one YY$ & $\one ZY$ & $\cx{X\one\one}$ & $\cx{X\one X}$ & $X\one Z$ & $XXY$ & $XYY$ & $XZY$ & $Y\one Y$ & $YX\one$ & $\cx{\one X\one}$ \\
    \cline{1-1}
    $YXX$ & $YXZ$ & $YY\one$ & $YYX$ & $YYZ$ & $YZ\one$ & $YZX$ & $YZZ$ & $\cz{Z\one\one}$ & $Z\one X$ & $\cz{Z\one Z}$ & $ZXY$ & $ZYY$ & $ZZY$ & $\cz{\one Z\one}$ \\
    \hline

\multicolumn{1}{|c|}{12} & $\cx{\one\one X}$ & $\cz{\one\one Z}$ & $\one XY$ & $\one YY$ & $\one ZY$ & $X\one Y$ & $\cx{XX\one}$ & $\cx{XXX}$ & $XXZ$ & $XY\one$ & $XYX$ & $XYZ$ & $XZ\one$ & $\cx{\one X\one}$ \\
    \cline{1-1}
    $XZX$ & $XZZ$ & $Y\one\one$ & $Y\one X$ & $Y\one Z$ & $YXY$ & $YYY$ & $YZY$ & $\cz{Z\one\one}$ & $Z\one X$ & $\cz{Z\one Z}$ & $ZXY$ & $ZYY$ & $ZZY$ & $\cz{ZZ\one}$ \\
    \hline

\multicolumn{1}{|c|}{13} & $\one\one Y$ & $\cz{\one\one Z}$ & $\cx{\one X\one}$ & $\one XY$ & $\one XZ$ & $\one Y\one$ & $\one YY$ & $\one YZ$ & $\one ZX$ & $\cx{X\one X}$ & $\cx{XXX}$ & $XYX$ & $XZ\one$ & $\cx{XX\one}$ \\
    \cline{1-1}
    $XZY$ & $XZZ$ & $Y\one X$ & $YXX$ & $YYX$ & $YZ\one$ & $YZY$ & $YZZ$ & $Z\one X$ & $ZXX$ & $ZYX$ & $\cz{ZZ\one}$ & $ZZY$ & $\cz{ZZZ}$ & $\cz{Z\one Z}$ \\
    \hline

\multicolumn{1}{|c|}{14} & $\one\one Y$ & $\cz{\one\one Z}$ & $\cx{\one X\one}$ & $\one XY$ & $\one XZ$ & $\one YX$ & $\cz{\one Z\one}$ & $\one ZY$ & $\cz{\one ZZ}$ & $\cx{X\one X}$ & $\cx{XXX}$ & $XY\one$ & $XYY$ & $\cx{X\one\one}$ \\
    \cline{1-1}
    $XYZ$ & $XZX$ & $Y\one X$ & $YXX$ & $YY\one$ & $YYY$ & $YYZ$ & $YZX$ & $Z\one X$ & $ZXX$ & $ZY\one$ & $ZYY$ & $ZYZ$ & $ZZX$ & $\cz{Z\one Z}$ \\
    \hline

\multicolumn{1}{|c|}{15} & $\one\one Y$ & $\cz{\one\one Z}$ & $\cx{\one XX}$ & $\one Y\one$ & $\one YY$ & $\one YZ$ & $\cz{\one Z\one}$ & $\one ZY$ & $\cz{\one ZZ}$ & $\cx{X\one X}$ & $\cx{XX\one}$ & $XXY$ & $XXZ$ & $\cx{X\one\one}$ \\
    \cline{1-1}
    $XYX$ & $XZX$ & $Y\one X$ & $YX\one$ & $YXY$ & $YXZ$ & $YYX$ & $YZX$ & $Z\one X$ & $ZX\one$ & $ZXY$ & $ZXZ$ & $ZYX$ & $ZZX$ & $\cz{ZZZ}$ \\
    \hline

\multicolumn{1}{|c|}{16} & $\one\one Y$ & $\cz{\one\one Z}$ & $\cx{\one XX}$ & $\one YX$ & $\one ZX$ & $\cx{X\one\one}$ & $X\one Y$ & $X\one Z$ & $\cx{XXX}$ & $XYX$ & $XZX$ & $Y\one\one$ & $Y\one Y$ & $\cx{XX\one}$ \\
    \cline{1-1}
    $Y\one Z$ & $YXX$ & $YYX$ & $YZX$ & $Z\one X$ & $ZX\one$ & $ZXY$ & $ZXZ$ & $ZY\one$ & $ZYY$ & $ZYZ$ & $\cz{ZZ\one}$ & $ZZY$ & $\cz{ZZZ}$ & $\cz{\one ZZ}$ \\
    \hline

\multicolumn{1}{|c|}{17} & $\one\one Y$ & $\cz{\one\one Z}$ & $\cx{\one XX}$ & $\one YX$ & $\one ZX$ & $\cx{X\one\one}$ & $X\one Y$ & $X\one Z$ & $\cx{XXX}$ & $XYX$ & $XZX$ & $Y\one X$ & $YX\one$ & $\cx{\one X\one}$ \\
    \cline{1-1}
    $YXY$ & $YXZ$ & $YY\one$ & $YYY$ & $YYZ$ & $YZ\one$ & $YZY$ & $YZZ$ & $\cz{Z\one\one}$ & $Z\one Y$ & $\cz{Z\one Z}$ & $ZXX$ & $ZYX$ & $ZZX$ & $\cz{\one ZZ}$ \\
    \hline

\multicolumn{1}{|c|}{18} & $\one\one Y$ & $\cz{\one\one Z}$ & $\cx{\one XX}$ & $\one YX$ & $\one ZX$ & $\cx{X\one X}$ & $\cx{XX\one}$ & $XXY$ & $XXZ$ & $XY\one$ & $XYY$ & $XYZ$ & $XZ\one$ & $\cx{\one X\one}$ \\
    \cline{1-1}
    $XZY$ & $XZZ$ & $Y\one\one$ & $Y\one Y$ & $Y\one Z$ & $YXX$ & $YYX$ & $YZX$ & $\cz{Z\one\one}$ & $Z\one Y$ & $\cz{Z\one Z}$ & $ZXX$ & $ZYX$ & $ZZX$ & $\cz{ZZZ}$ \\
    \hline

\multicolumn{1}{|c|}{19} & $\cx{\one X\one}$ & $\one Y\one$ & $\one ZX$ & $\one ZY$ & $\cz{\one ZZ}$ & $\cx{X\one\one}$ & $\cx{XX\one}$ & $XY\one$ & $XZX$ & $XZY$ & $XZZ$ & $Y\one\one$ & $YX\one$ & $\cx{XXX}$ \\
    \cline{1-1}
    $YY\one$ & $YZX$ & $YZY$ & $YZZ$ & $Z\one X$ & $Z\one Y$ & $\cz{Z\one Z}$ & $ZXX$ & $ZXY$ & $ZXZ$ & $ZYX$ & $ZYY$ & $ZYZ$ & $\cz{ZZ\one}$ & $\cz{\one\one Z}$ \\
    \hline

\multicolumn{1}{|c|}{20} & $\cx{\one X\one}$ & $\one Y\one$ & $\one ZX$ & $\one ZY$ & $\cz{\one ZZ}$ & $\cx{X\one\one}$ & $\cx{XX\one}$ & $XY\one$ & $XZX$ & $XZY$ & $XZZ$ & $Y\one X$ & $Y\one Y$ & $\cx{\one XX}$ \\
    \cline{1-1}
    $Y\one Z$ & $YXX$ & $YXY$ & $YXZ$ & $YYX$ & $YYY$ & $YYZ$ & $YZ\one$ & $\cz{Z\one\one}$ & $ZX\one$ & $ZY\one$ & $ZZX$ & $ZZY$ & $\cz{ZZZ}$ & $\cz{\one\one Z}$ \\
    \hline

\multicolumn{1}{|c|}{21} & $\cx{\one X\one}$ & $\one Y\one$ & $\one ZX$ & $\one ZY$ & $\cz{\one ZZ}$ & $\cx{X\one X}$ & $X\one Y$ & $X\one Z$ & $\cx{XXX}$ & $XXY$ & $XXZ$ & $XYX$ & $XYY$ & $\cx{\one XX}$ \\
    \cline{1-1}
    $XYZ$ & $XZ\one$ & $Y\one\one$ & $YX\one$ & $YY\one$ & $YZX$ & $YZY$ & $YZZ$ & $\cz{Z\one\one}$ & $ZX\one$ & $ZY\one$ & $ZZX$ & $ZZY$ & $\cz{ZZZ}$ & $\cz{Z\one Z}$ \\
    \hline

\multicolumn{1}{|c|}{22} & $\cx{\one X\one}$ & $\one YX$ & $\one YY$ & $\one YZ$ & $\cz{\one Z\one}$ & $\cx{X\one\one}$ & $\cx{XX\one}$ & $XYX$ & $XYY$ & $XYZ$ & $XZ\one$ & $Y\one\one$ & $YX\one$ & $\cx{X\one X}$ \\
    \cline{1-1}
    $YYX$ & $YYY$ & $YYZ$ & $YZ\one$ & $Z\one X$ & $Z\one Y$ & $\cz{Z\one Z}$ & $ZXX$ & $ZXY$ & $ZXZ$ & $ZY\one$ & $ZZX$ & $ZZY$ & $\cz{ZZZ}$ & $\cz{\one\one Z}$ \\
    \hline

\multicolumn{1}{|c|}{23} & $\cx{\one X\one}$ & $\one YX$ & $\one YY$ & $\one YZ$ & $\cz{\one Z\one}$ & $\cx{X\one\one}$ & $\cx{XX\one}$ & $XYX$ & $XYY$ & $XYZ$ & $XZ\one$ & $Y\one X$ & $Y\one Y$ & $\cx{\one\one X}$ \\
    \cline{1-1}
    $Y\one Z$ & $YXX$ & $YXY$ & $YXZ$ & $YY\one$ & $YZX$ & $YZY$ & $YZZ$ & $\cz{Z\one\one}$ & $ZX\one$ & $ZYX$ & $ZYY$ & $ZYZ$ & $\cz{ZZ\one}$ & $\cz{\one\one Z}$ \\
    \hline

\multicolumn{1}{|c|}{24} & $\cx{\one X\one}$ & $\one YX$ & $\one YY$ & $\one YZ$ & $\cz{\one Z\one}$ & $\cx{X\one X}$ & $X\one Y$ & $X\one Z$ & $\cx{XXX}$ & $XXY$ & $XXZ$ & $XY\one$ & $XZX$ & $\cx{\one\one X}$ \\
    \cline{1-1}
    $XZY$ & $XZZ$ & $Y\one\one$ & $YX\one$ & $YYX$ & $YYY$ & $YYZ$ & $YZ\one$ & $\cz{Z\one\one}$ & $ZX\one$ & $ZYX$ & $ZYY$ & $ZYZ$ & $\cz{ZZ\one}$ & $\cz{Z\one Z}$ \\
    \hline

\multicolumn{1}{|c|}{25} & $\cx{\one XX}$ & $\one XY$ & $\one XZ$ & $\one Y\one$ & $\cz{\one Z\one}$ & $\cx{X\one\one}$ & $\cx{XXX}$ & $XXY$ & $XXZ$ & $XY\one$ & $XZ\one$ & $Y\one\one$ & $YXX$ & $\cx{X\one X}$ \\
    \cline{1-1}
    $YXY$ & $YXZ$ & $YY\one$ & $YZ\one$ & $Z\one X$ & $Z\one Y$ & $\cz{Z\one Z}$ & $ZX\one$ & $ZYX$ & $ZYY$ & $ZYZ$ & $ZZX$ & $ZZY$ & $\cz{ZZZ}$ & $\cz{\one ZZ}$ \\
    \hline

\multicolumn{1}{|c|}{26} & $\cx{\one XX}$ & $\one XY$ & $\one XZ$ & $\one Y\one$ & $\cz{\one Z\one}$ & $\cx{X\one\one}$ & $\cx{XXX}$ & $XXY$ & $XXZ$ & $XY\one$ & $XZ\one$ & $Y\one X$ & $Y\one Y$ & $\cx{\one\one X}$ \\
    \cline{1-1}
    $Y\one Z$ & $YX\one$ & $YYX$ & $YYY$ & $YYZ$ & $YZX$ & $YZY$ & $YZZ$ & $\cz{Z\one\one}$ & $ZXX$ & $ZXY$ & $ZXZ$ & $ZY\one$ & $\cz{ZZ\one}$ & $\cz{\one ZZ}$ \\
    \hline

    \multicolumn{1}{|c|}{27} & $\cx{\one XX}$ & $\one XY$ & $\one XZ$ & $\one Y\one$ & $\cz{\one Z\one}$ & $\cx{X\one X}$ & $X\one Y$ & $X\one Z$ & $\cx{XX\one}$ & $XYX$ & $XYY$ & $XYZ$ & $XZX$ & \cx{$\one\one X$} \\
    \cline{1-1}
    $XZY$ & $XZZ$ & $Y\one\one$ & $YXX$ & $YXY$ & $YXZ$ & $YY\one$ & $YZ\one$ & $\cz{Z\one\one}$ & $ZXX$ & $ZXY$ & $ZXZ$ & $ZY\one$ & $\cz{ZZ\one}$ & \cz{$ZZZ$} \\
    \hline
    \multicolumn{1}{|c|}{28} & $\cx{\one XX}$ & $\one XY$ & $\one XZ$ & $\one YX$ & $\one YY$ & $\one YZ$ & $\one ZX$ & $\one ZY$ & $\cz{\one ZZ}$ & $\cx{X\one X}$ & $X\one Y$ & $X\one Z$ & $\cx{XX\one}$ & \cx{$XXX$} \\
    \cline{1-1}
    $XY\one$ & $XZ\one$ & $Y\one X$ & $Y\one Y$ & $Y\one Z$ & $YX\one$ & $YY\one$ & $YZ\one$ & $Z\one X$ & $Z\one Y$ & $\cz{Z\one Z}$ & $ZX\one$ & $ZY\one$ & $\cz{ZZ\one}$ & \cz{$ZZZ$} \\
    \hline
   \end{tabular}%
   }
   }
\end{figure}

\end{document}